\documentclass[10pt, conference, letterpaper]{IEEEtran}
\IEEEoverridecommandlockouts

\usepackage{cite}
\usepackage{amsmath,amssymb,amsfonts}
\usepackage{algorithmic}
\usepackage{graphicx}
\usepackage{textcomp}
\usepackage{xcolor}
\usepackage[font=small]{caption}

\usepackage{subfigure}
\usepackage[ruled,vlined,linesnumbered, noend]{algorithm2e}
\graphicspath{ {figures/} }
\usepackage[export]{adjustbox}
\usepackage{multirow}
\usepackage{mathtools}
\usepackage{amsthm}

\usepackage{caption}

\newtheorem{theorem}{Theorem}

\def\BibTeX{{\rm B\kern-.05em{\sc i\kern-.025em b}\kern-.08em
		T\kern-.1667em\lower.7ex\hbox{E}\kern-.125emX}}
\begin{document}

\title{Deep Learning on Mobile Devices Through Neural Processing Units and Edge Computing}




\author{
\IEEEauthorblockN{Tianxiang Tan and Guohong Cao}
\IEEEauthorblockA{
	Department of Computer Science and Engineering \\
	The Pennsylvania State University\\
	Email: \{txt51, gxc27\}@psu.edu}
}

\maketitle

\begin{abstract}	
Deep Neural Network (DNN) is becoming adopted for video analytics on mobile devices. To reduce the delay of running DNNs, many mobile devices are equipped with Neural Processing Units (NPU). However, due to the resource limitations of NPU, these DNNs have to be compressed to increase the processing speed at the cost of accuracy. To address the low accuracy problem, we propose a Confidence Based Offloading (CBO) framework for deep learning video analytics. The major challenge is to determine when to return the NPU classification result based on the confidence level of running the DNN, and when to offload the video frames to the server for further processing to increase the accuracy. We first identify the problem of using existing confidence scores to make offloading decisions, and propose confidence score calibration techniques to improve the performance. Then, we formulate the CBO problem where the goal is to maximize accuracy under some time constraint, and propose an adaptive solution that determines which frames to offload at what resolution based on the confidence score and the network condition. Through real implementations and extensive evaluations, we demonstrate that the proposed solution can significantly outperform other approaches. 

\end{abstract}


\section{Introduction}

Deep Neural Networks (DNN) have been successfully applied to various computer vision and natural language processing problems.
Recently, many applications based on DNNs have been developed to provide more intelligent video analytics.
For example, some drones such as DJI Mavic Pro can recognize and follow a target based on video analytics;
law enforcement officers can use smart glasses to identify suspects \cite{bbc-news}.  
In these applications, only lightweight DNNs can be run locally and their accuracy is much lower than advanced DNNs.
Although advanced DNNs can provide us with better results, they also suffer from high computational overhead which means long delay and more energy consumption when running on mobile devices.

In recent years, many companies such as Huawei, Qualcomm, and Samsung have developed dedicated {\em Neural Processing Units (NPUs)} for mobile devices, which can process AI features. 
With NPU, the running time of these DNNs can be significantly reduced.
For example, on HUAWEI mate 10 pro, running AlexNet (a DNN) \cite{alex-nips12} on NPU is 30 times faster than running it on CPU.
Although GPUs on mobile devices can also be used to run DNNs, they are not as powerful as those running on desktops and servers. 
Compared to GPUs on mobile devices, NPUs are much faster and more energy efficient \cite{npu-gpu-compare}.
However, there are some fundamental limitations with NPU.
The most significant limitation of NPU is the precision of the floating-point numbers.
NPU uses 16 bits or 8 bits to represent the floating-point numbers instead of 32 bits in CPU. 
As a result, it runs DNNs much faster but less accurate compared to CPU.
Moreover, NPU has its own memory which is usually too small for advanced DNNs.
The advanced DNNs have to be compressed in order to be loaded and executed by NPU at the cost of accuracy.

To address the low accuracy problem on NPU, 
mobile devices can offload the data to the edge server and let the edge server run the DNNs.
Since the server has more computational power, high accuracy can be achieved by processing the data with
the most advanced DNNs in a short time. However, when the network condition is poor or when the data size is large which is usually true in video analytics, offloading may take longer time to transmit the data.
There is a tradeoff between the offloading based approach
and the NPU based approach. The offloading based approach
has good accuracy, but has longer delay under poor
network condition. On the other hand, the NPU based approach is faster, but with less accuracy. 
Recent research (e.g., FastVA \cite{tan-infocom2020}) tries to combine these two approaches
for video analytics on mobile devices. To maximize accuracy under some time constraint, 
FastVA determines which frames should be processed locally on NPU and
which frames should be offloaded to the edge server.
However, existing research treats DNN as a black box and never considers exploring insights about 
running DNNs on NPU to improve performance. 

We address this problem by considering the confidence level of running DNNs on NPU. 
The confidence level can be represented by the confidence score, which is computed based on the output of the DNN. 
If the confidence score is higher than a threshold, the classification result on NPU is most likely accurate and can be directly used; otherwise, the data is offloaded for further processing to improve the accuracy. 
Similar ideas have been studied in existing work although for different purpose. For example, in \cite{teerapittayanon-icdcs17, wang-infocom2020}, confidence score is leveraged to reduce the processing delay by early exit; i.e., returning the classification results without running all DNN layers. 
However, there are two unique challenges for applying this idea to video analytics on mobile devices with NPU. 
First, the confidence score of many advanced DNNs cannot accurately estimate the classification results, and then may not be effective for making offloading decisions. 
Second, existing research relies on fixed confidence score threshold to determine when to offload. However, in video analytics, the decision can be affected by many factors such as the confidence score, the network condition, the video frame resolution, and thus fixed threshold will not work well. 

In this paper, we address these challenges by proposing a Confidence Based Offloading (CBO) framework for deep learning video analytics on mobile devices with NPU. 
To address the first challenge, we propose techniques to calibrate the confidence score so that it can accurately reflect the correctness of the classification results on NPU.
To address the second challenge, we formulate the CBO problem, where the goal is to 
maximize accuracy under some time constraint. To achieve this goal, the confidence score threshold is adaptively adjusted based on the network condition, confidence score and the selected frame resolution, before being used for determining if the video frame should be offloaded for further processing. 

Our contributions can be summarized as follows.

\begin{itemize}
	
	\item
	To address the accuracy loss problem in deep learning video analytics on NPU based mobile devices, 
	we study confidence based offloading and propose confidence score calibration techniques to improve the performance. 
	
	\item We formulate the CBO problem, and propose an optimal solution. Since the optimal solution relies on the perfect knowledge of all future video frames, we further propose an online algorithm 
	to relax this assumption. 
	
	\item  
	Through real implementations and extensive evaluations, we demonstrate that the proposed solution can significantly outperform other approaches. 
	
\end{itemize}

The rest of the paper is organized as follows.
Section \ref{sec:preliminary} presents the background and motivation.
In Section \ref{sec:quantify-conf}, we evaluate the effectiveness of CBO and propose calibration techniques to improve performance.
In Section \ref{sec:max-acc}, we present our CBO framework for video analytics. 
Section \ref{sec:evaluation} presents the evaluation results and Section \ref{sec:related-work} reviews related work.
Section \ref{sec:conclusion} concludes the paper.

\section{Preliminary} \label{sec:preliminary}


\subsection{Characteristics of NPU}

\begin{figure}
	\centering
	\subfigure[Processing time comparison]{\includegraphics[width=0.48\linewidth]{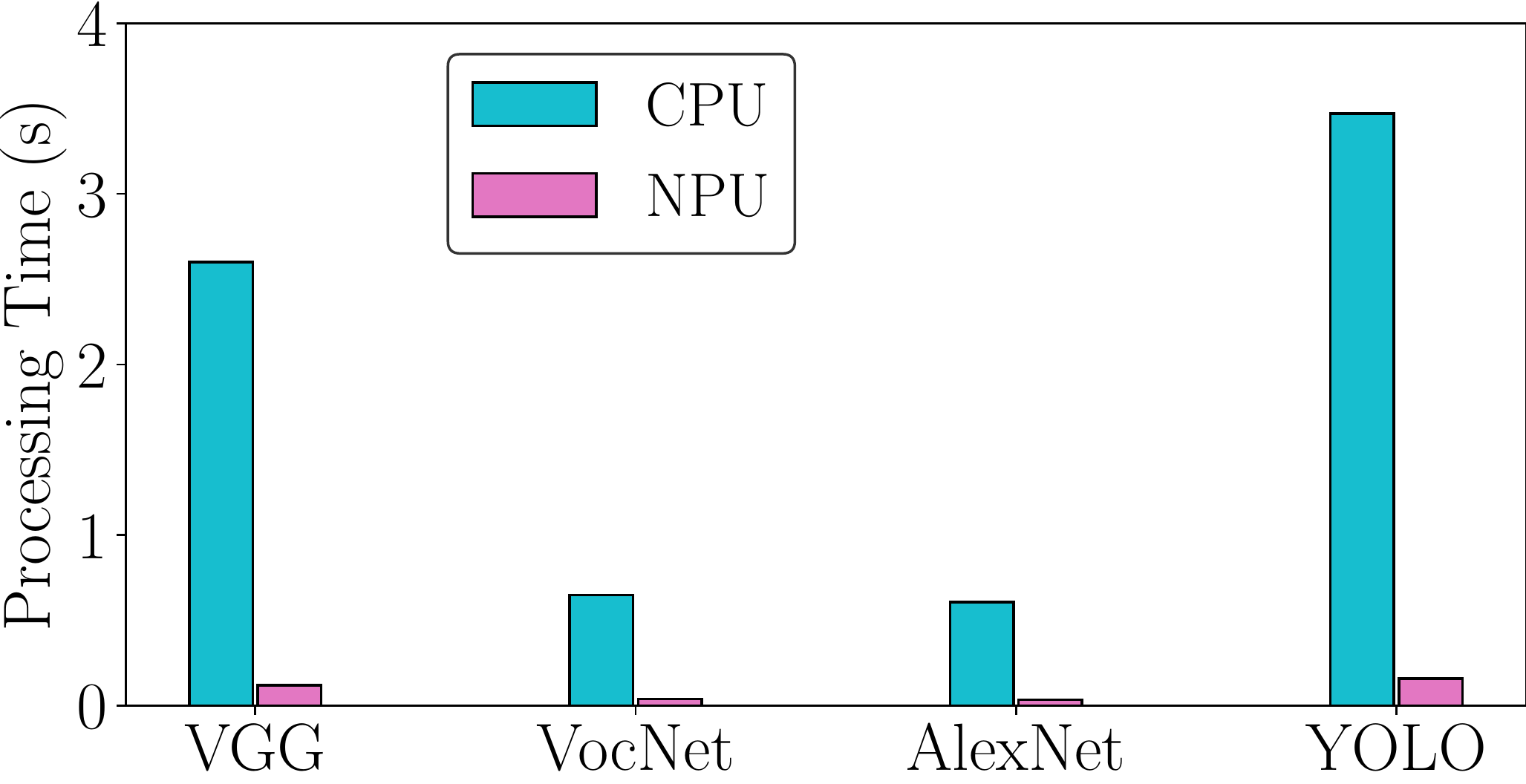}}
	\,
	\subfigure[Accuracy comparison]{\includegraphics[width=0.48\linewidth]{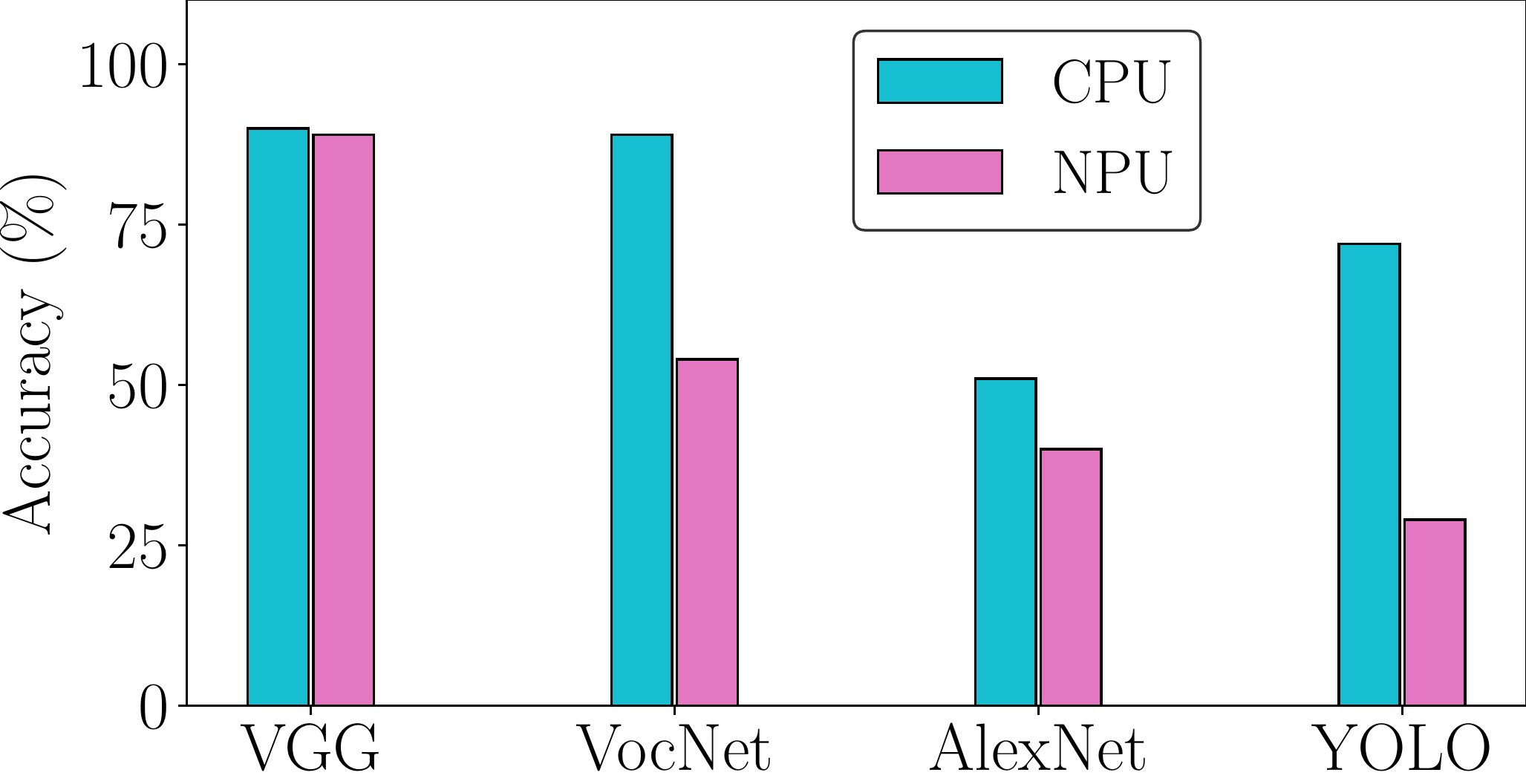}}
	\vspace{-0.5em}
	\caption{Performance comparisons of running DNNs on NPU and CPU.}
	\label{fig:NPU-profile}
	\vspace{-2em}
\end{figure}


To have a better understanding of NPU, we conducted some experiments on HUAWEI mate 10 pro, using the following four DNN models: 1) VGG \cite{parkhi-bmvc15} with the LFW dataset \cite{lfw-dataset},
2) VocNet \cite{lapuschkin-cvpr16} with 4000 images from the VOC dataset \cite{everingham-ijcv15},
3) AlexNet \cite{alex-nips12} with the VOC dataset,
4) YOLO Small \cite{redmon-cvpr16} with 4000 images randomly chosen from the MS COCO dataset \cite{lin-eccv14}.
As shown in Figure \ref{fig:NPU-profile}(a), compared to CPU, running VGG, VocNet, AlexNet, Yolo Small on NPU can significantly reduce the processing time by 95\%; however, this may be at the cost of accuracy loss as 
illustrated in Figure \ref{fig:NPU-profile}(b). 
Specifically, compared to CPU, using NPU has similar accuracy when running VGG, 30\% accuracy loss when running VocNet, 11\% accuracy loss for AlexNet, and the F1-score drops to 0.3 for YOLO Small.

The accuracy loss is mainly because NPU can only support FP16 operations and store the intermediate result of each layer using FP16. Running DNNs with FP16 can save memory and reduce the processing time, but it may reduce the accuracy due to the numerical instability caused by floating point overflow or underflow.

The amount of accuracy loss is related to the DNN model.
In VGG, 
the extracted feature vectors are compared, and they represent the same person if the similarity is above a predefined threshold.
Although NPU introduced error may change some values, the relationship between the similarity and the threshold will remain and thus keep the same level of accuracy.
However, AlexNet, VocNet and Yolo Small use more information in the feature vectors to identify or locate multiple objects in the images.
Each value in the feature vector represents the category, the location or the size of an object and a small error introduced by NPU can change the prediction completely.
As a result, they have lower accuracy when running on NPU.

As shown in Figure \ref{fig:NPU-profile}, NPU runs much faster than CPU. 
It is ideal for running some DNNs such as VGG, where using NPU can significantly reduce the processing time while maintaining high accuracy. 
However, it may not be the best choice for running DNNs such as VocNet and Yolo small, due to the high accuracy loss. 
In this paper, we address this problem by proposing techniques to improve its accuracy.

\subsection {Motivation}


To address the accuracy loss problem, one solution is to offload the video frames to the edge server or cloud for processing.
However, transmitting all video frames to the server may take much longer time and may violate the delay constraint, compared to processing them locally on NPU, especially when the wireless bandwidth is limited.  
Instead of processing video frames only on NPU or only through offloading, we take advantage of both approaches by exploring insights about running DNNs on NPU.  

Although running DNNs on NPU has lower accuracy in general, the accuracy of running DNNs for some specific objects can be much higher, due to the skewed accuracy of the DNNs.
This is because DNNs are becoming more and more complex, with many parameters and layers designed to achieve high accuracy for general inference scenarios. These DNNs have lots of redundancy especially when the inference scenario is simple; i.e.,  
even when running on NPU with lower precision parameter and simple model structure, the compressed DNN can still perform well for some simple scenarios. 

\begin{figure}[t]
	\centering
	\includegraphics[width=0.65\linewidth]{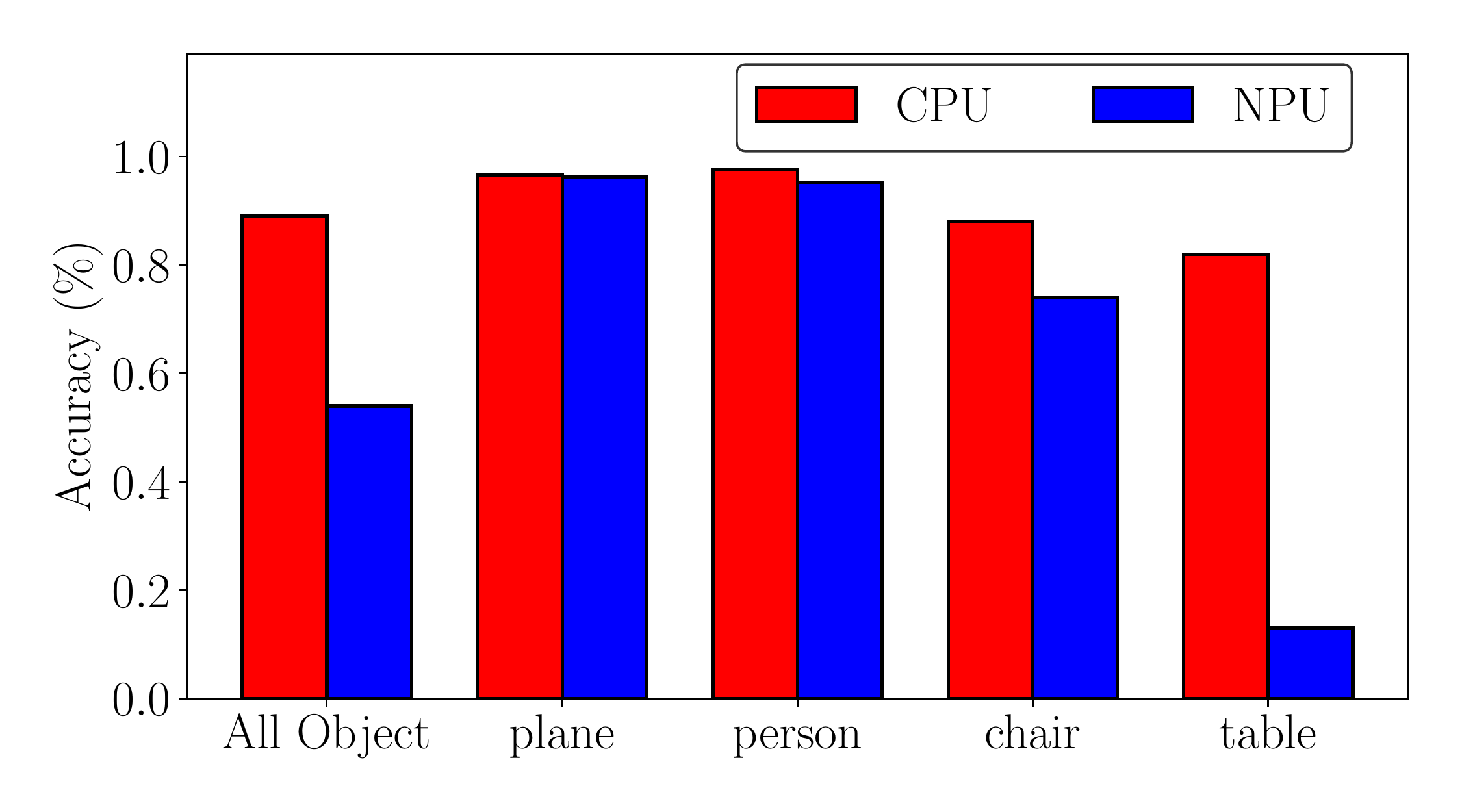}
	\vspace{-1em}
	\caption{Accuracy for different objects with VocNet}
	\label{fig:object-acc}
	\vspace{-2.3em}
\end{figure}

Figure \ref{fig:object-acc} shows the accuracy of running VocNet on CPU and NPU for different objects.
When considering all objects, the accuracy of running VocNet on NPU is 0.54, which is lower than that on CPU (0.89).
When considering specific objects, VocNet on NPU performs much better for recognizing airplanes with accuracy of 0.96, which is similar to the accuracy on CPU.
This is because there is much higher visual difference between the background (i.e., blue sky) and the airplane.  
Even though the VocNet has been compressed to be run on NPU, with lower precision parameters and simpler model structure, it can still recognize airplanes accurately.
On the other hand, VocNet performs poorly on NPU for recognizing tables, with accuracy of 0.1. This is because  
most tables in the data set have much complex background, and the visual difference between the background and table is much smaller. Then, a more complex and advanced DNN (i.e., the original DNN) is needed to classify these objects correctly.

To consider both cases, two different DNNs should be applied, a compressed DNN to process simple cases and the original DNN to process more complex cases. Due to the resource limitation of NPU, the compressed DNN should be run on NPU and the original DNN should be run on the edge server or cloud. 
Then, images such as airplane should be processed with local NPU and other images such as tables should be offloaded to the server for further processing. 
However, the system does not know whether an image is an airplane or table before hand. 
Thus, we propose the following solution. First, the image is processed locally on NPU, which has negligible delay. 
If the image can be recognized with a high level of confidence, the classification result is returned; otherwise, the image is offloaded for further processing. 
Then, most images can be processed locally with high accuracy and low delay, and only images with low accuracy are offloaded. 
The key component of this solution is how to quantify the confidence level of running DNNs on NPU, which will be discussed in the next section.

\begin{figure*}[t]
	\centering
	\begin{minipage}[t]{0.235\textwidth}
		\centering
		\includegraphics[width=\linewidth]{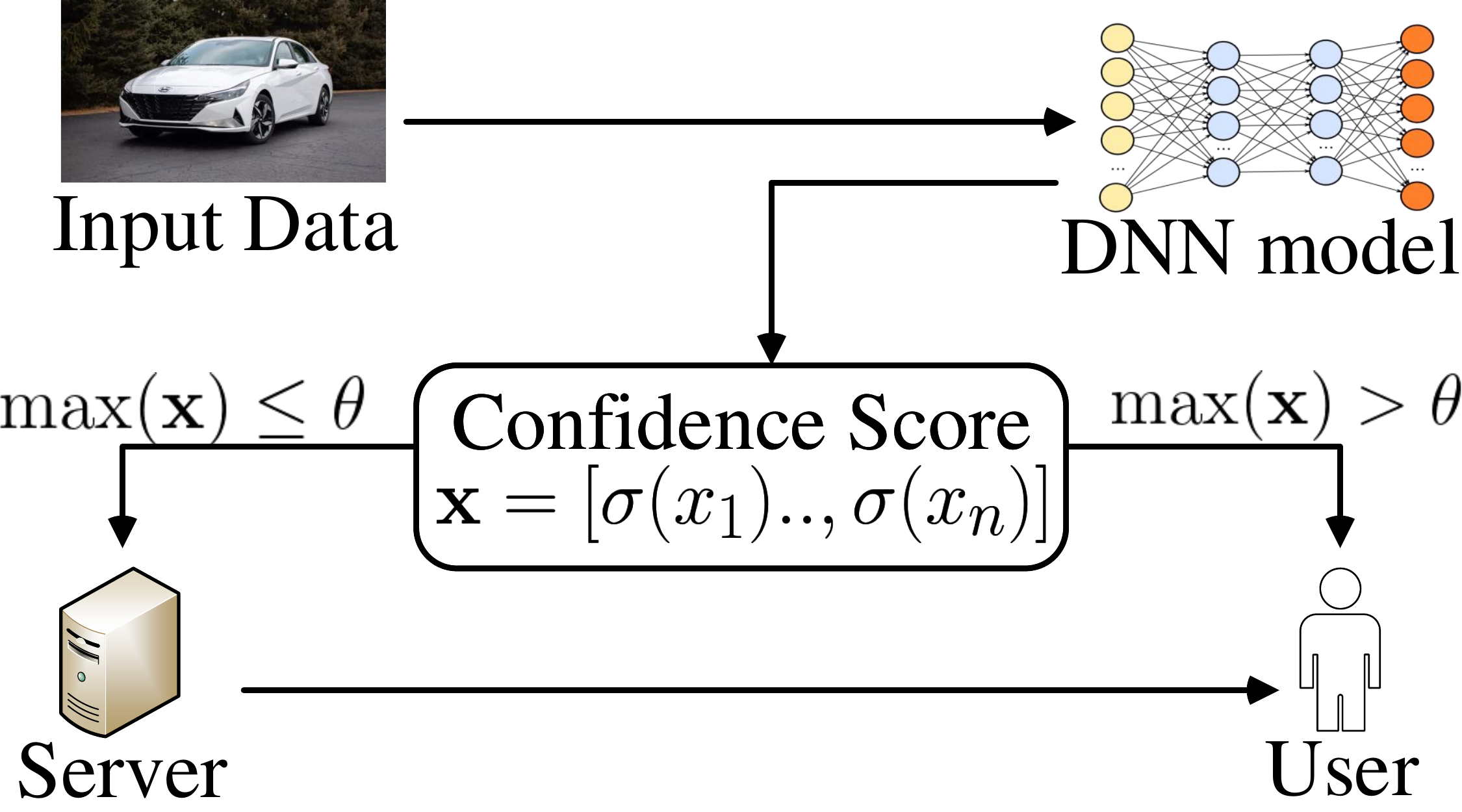}
		\caption{Making offloading decisions based on confidence score.}\label{fig:offload-conf-score}
	\end{minipage}
	\,
	\begin{minipage}[t]{0.235\textwidth}
		\centering
		\includegraphics[width=\linewidth]{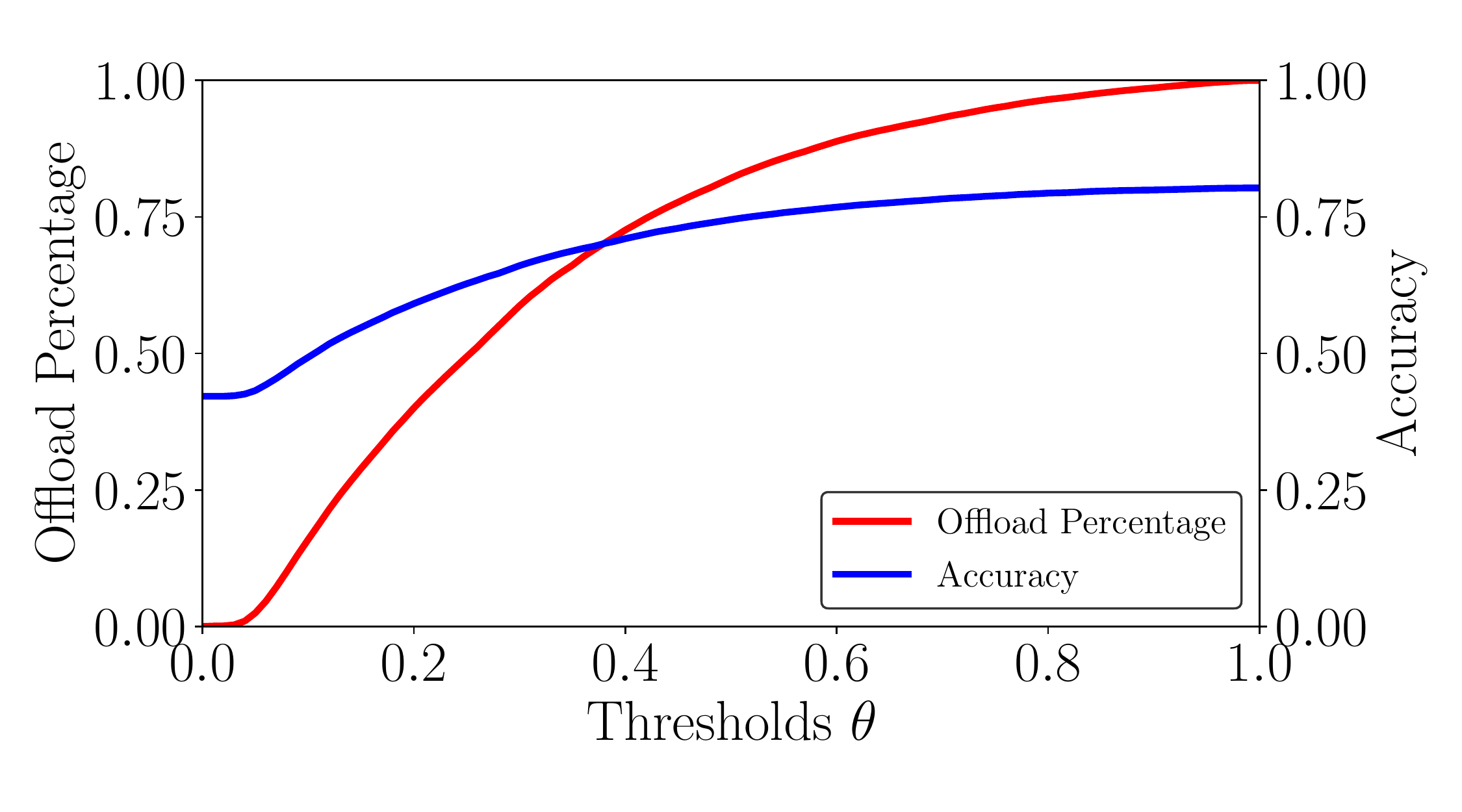}
		\caption{The effects of confidence score threshold $\theta$.}\label{fig:org-conf-score}
	\end{minipage}
	\,
	\begin{minipage}[t]{0.235\textwidth}
		\centering
		\includegraphics[width=\linewidth]{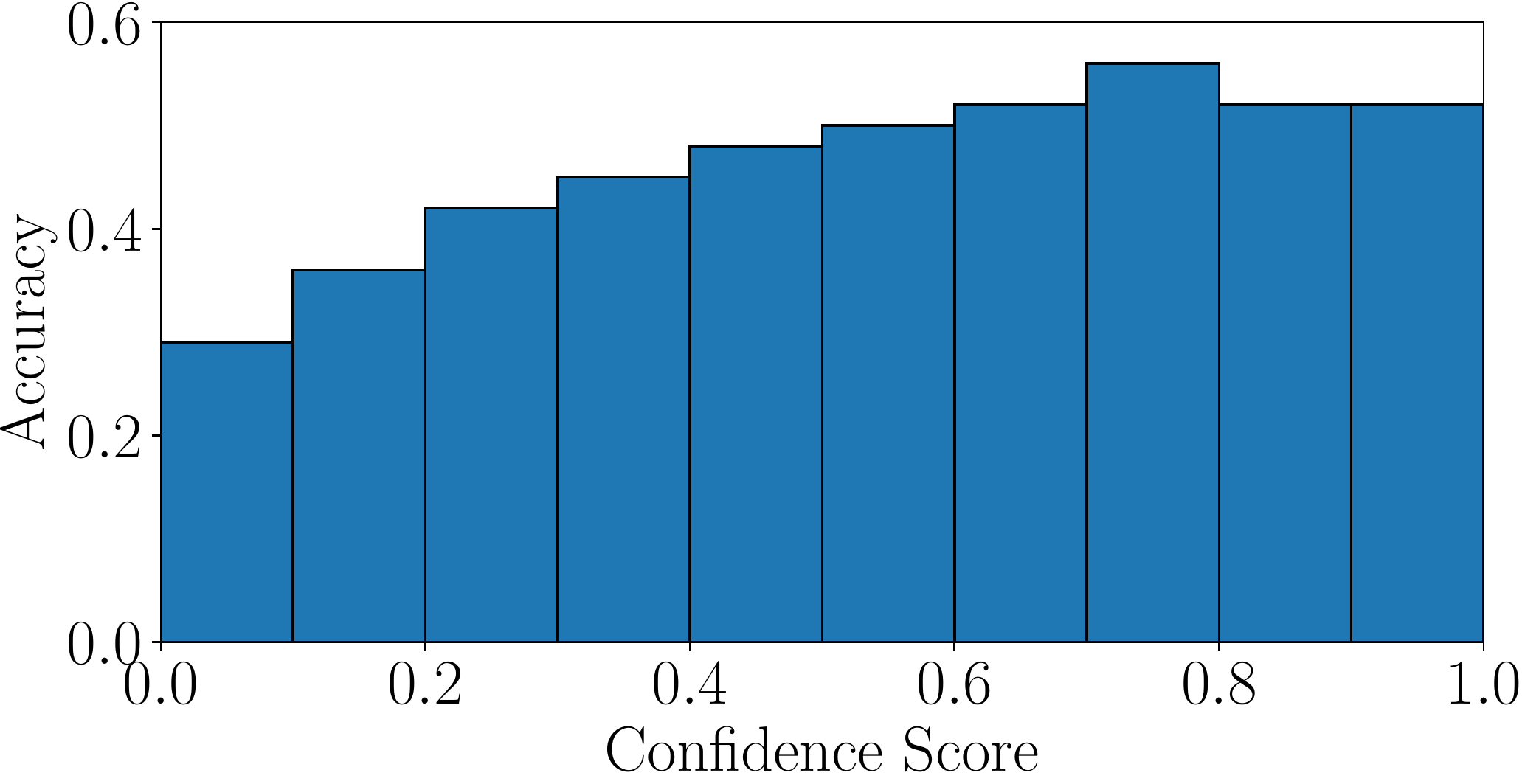}
		\caption{Accuracy vs. confidence score.}\label{fig:conf-acc-uncal}
	\end{minipage}
	\,
	\begin{minipage}[t]{0.235\textwidth}
		\centering
		\includegraphics[width=\linewidth]{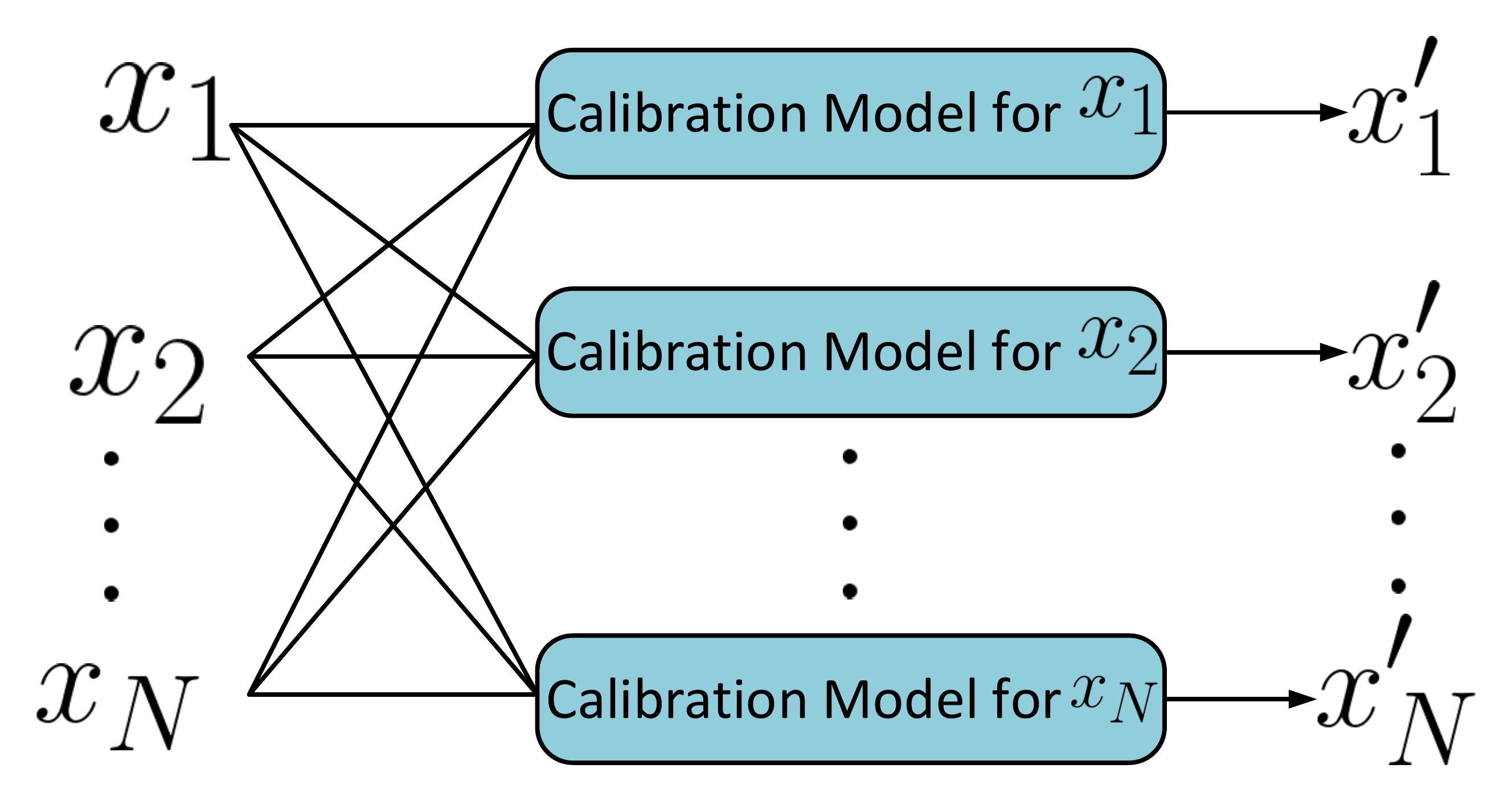}
		\caption{Calibration model training.}\label{fig:cal-model-train}
	\end{minipage}
	\vspace{-1em}
\end{figure*}

\section{Quantifying the confidence level} \label{sec:quantify-conf}

The result of the DNN is represented by a feature vector which can be extracted from the last layer of the DNN.
To quantify the confidence level of the DNN, the feature vector is used to calculate the confidence score.
In this section, we first present an offloading framework based on the calculated confidence score, and 
and then present confidence score calibration techniques to improve its accuracy.

\subsection{Confidence Based Offloading (CBO)}


Figure \ref{fig:offload-conf-score} illustrates the basic idea of CBO. 
To determine which images should be offloaded, confidence score is used to predict the correctness of the result.
The confidence score can be computed based on the feature vector produced by the DNN.
Let $\mathbf{x} = [x_1, x_2, x_3, \ldots, x_N]$ denote the extracted feature vector, where 
each element $x_i$ represents the probability that object $i$ appears in the image.
Normally, the values of these elements are not normalized and the confidence score is computed as $\max \sigma(x_i)$,
where $\sigma(x_i)$ is the softmax function and it is defined as $\sigma(x_i)=\frac{e^{x_i}}{\sum_{k=1}^N e^{x_k}}$.
As shown in Figure \ref{fig:offload-conf-score}, when the confidence score is larger than a threshold $\theta$, 
the result is returned; otherwise, the image is offloaded for further processing. 

We conducted experiments to evaluate the effectiveness of CBO.
The AlexNet and ResNet-152 which are trained based on the ImageNet dataset, are used to process videos randomly selected from the FCVID dataset, which includes thousands of real-world videos.
In the experiment, AlexNet is run on NPU and ResNet-152 is run on the server.
The offload percentage is defined as $\frac{n^o}{n}$, where $n^o$ is the number of offloaded frames and $n$ is the total number of video frames.

The results are shown in Figure \ref{fig:org-conf-score}.
When $\theta$ is 0, all frames are only processed by NPU.
When $\theta$ is 1, all frames are offloaded to the server. 
As shown in the figure, as $\theta$ increases from 0 to 1, more frames are offloaded to the server and the accuracy increases from 0.42 to 0.81. However, this is not as expected due to the following reason. 
When the accuracy requirement is 0.8, the threshold has to be above 0.8 at which the offloaded traffic is above 90\%. Then, CBO does not save any bandwidth. 

The reason of such poor performance is because the confidence score generated by the DNN on NPU 
cannot accurately estimate the correctness of the result. We conducted the following experiment.
With the same DNN and the dataset to identify the reason. In this experiment, 
the video frames are divided into 10 bins with 0.1 confidence interval. 
For example, the first bin includes the frames with confidence score from 0 to 0.1.
For each bin, we process the frames on NPU and calculate the accuracy.
The result is shown in Figure \ref{fig:conf-acc-uncal}.
As the confidence score increases from 0 to 1, the accuracy increases from 0.29 to 0.5.
For example, the accuracy is 0.42 for the frames with confidence score between 0.2 to 0.3.

CBO is based on the assumption that frames classified incorrectly on NPU will be offloaded for further processing. This can be easily achieved if the confidence score of running the DNN to process the frame is the same as its accuracy. Then, the classification result of processing frames with confidence score 0.9 is more likely to be correct 
than that with confidence score 0.5, and hence the accuracy can be improved by offloading frames with lower confidence score. 
However, as shown in Figure \ref{fig:conf-acc-uncal}, the accuracy remains to be 0.5 for frames with confidence score much higher than 0.5 (e.g., 0.9). Then, it is hard to use the confidence score to determine which frame should be offloaded. This explains why CBO performs poorly as shown in Figure \ref{fig:org-conf-score}.
Fortunately, there is nothing wrong with CBO. The problem is due to the fact that the confidence score cannot accurately estimate the correctness of the classification result, and it should be calibrated. 

\subsection{Confidence Score Calibration}

Although modern DNNs can achieve better accuracy in classification problems, the confidence scores produced by these DNNs are poorly calibrated \cite{guo-icml17}.
To address this issue, researchers have proposed many confidence score calibration techniques, following two different methods. 

\textit{Platt Calibration:}
Platt Calibration is a parametric method \cite{niculescu-ICML05} and its key idea is to train logistic regression models which can transform the confidence score into a calibrated one to provide better correctness estimation.
The logistic model can be defined as $P(y_i=1|\mathbf{x}) = \frac{1}{1 + e^{\mathbf{A} f(\mathbf{x}) + \mathbf{B}}}$, where $\mathbf{A, B}$ are the parameters that needs to be trained.

\textit{Isotonic Regression:}
Different from Platt Calibration, the Isotonic Regression is a common non-parametric calibration method \cite{zadrozny-sigkdd02}.
It learns a piecewise constant function $f$ to transform the confidence score into a calibrated one, which means $x'_i = f(x_i)$.
To train the Isotonic Regression model, $f$ is trained by minimizing the square loss function $\sum_{i=1}^{n} (f(x_i) - y_i)^2$.





Although these two methods use different machine learning models to calibrate the confidence score, they have the same training procedure for confidence score calibration, and this procedure is shown in Figure \ref{fig:cal-model-train}.
As shown in the figure, $N$ models are trained for calibrating the confidence score $x_1, x_2, \ldots, x_N$.
For each model $i$, it takes the feature vector $\mathbf{x} = [x_1, x_2, x_3, \ldots, x_N]$ as input and outputs a new confidence score $x'_i$.
When the $i^{th}$ model is being trained, the output variable $y_j$ is set to be 1 if object $i$ appears in the frame $I_j$.
$y_j$ is set to be 0 if $I_j$ does not include object $i$.

In our CBO, our goal is to choose a method which can calibrate the 
confidence score and make it match the accuracy of the DNN.
To achieve this goal, we use the following metrics to evaluate them and select the best. 

\begin{itemize}
	\item Expected Calibration Error (ECE): ECE measures the expectation difference between the confidence score and accuracy of the DNN. ECE is defined as
	$\sum_{i=1}^{10} \frac{B_i}{n} |A(B_i) - C(B_i)|$,	where $B_i$ ($1 \leq i \leq 10$) is the bins with 0.1 confidence interval. For example $B_1$ includes the frames with confidence score between 0 to 0.1. 
	For each bin $B_i$, we process the frames on NPU and calculate the accuracy $A(B_i)$ and the average confidence score $C(B_i)$.
	
	\item Maximum Calibration Error (MCE): MCE measures the worst-case difference between the confidence score and accuracy of the DNN. 
	MCE is defined as $\max_{1 \leq i \leq 10} |A(B_i) - C(B_i)|$.
\end{itemize}

\begin{table}[h]
	\centering
	\begin{tabular}{|l|l|l|}
		\hline
		Method               & ECE  & MCE  \\ \hline
		Uncalibrated         & 0.27 & 0.48 \\ \hline
		Platt Calibration        & 0.07 & 0.29 \\ \hline
		Isotonic Regression  & 0.16 & 0.41 \\ \hline
	\end{tabular}
	\vspace{-0.5em}
	\caption{Comparison of different calibration techniques.}
	\label{table:cal-comparison}
\end{table}

With AlexNet and the FCVID dataset, we perform an experiment to evaluate the performance of different calibration techniques. 
In the experiment, the video frames are divided into 10 bins with 0.1 confidence interval. 
For example, the first bin includes the frames with confidence score from 0 to 0.1.
For each bin, we process the frames on NPU and calculate the accuracy.
We use the images from ImageNet validation dataset and a subset of video frames from FCVID dataset as the training data for Platt Calibration and Isotonic Regression. 
The result is shown in Table \ref{table:cal-comparison}.
The uncalibrated confidence score has the worst performance among the three methods. 
Compared to the Isotonic Regression, Platt Calibration has lower ECE and MCE, which means its performance is better.
This is because the Isotonic Regression suffers from overfitting, where the model can learn the training data well but is not generalized to new data.
Therefore, Platt Calibration is used for calibrating the confidence score.

\begin{figure}[t]
	\centering
	\subfigure[Accuracy and offload traffic under different calibrated confidence score thresholds.]{\includegraphics[width=0.48\linewidth]{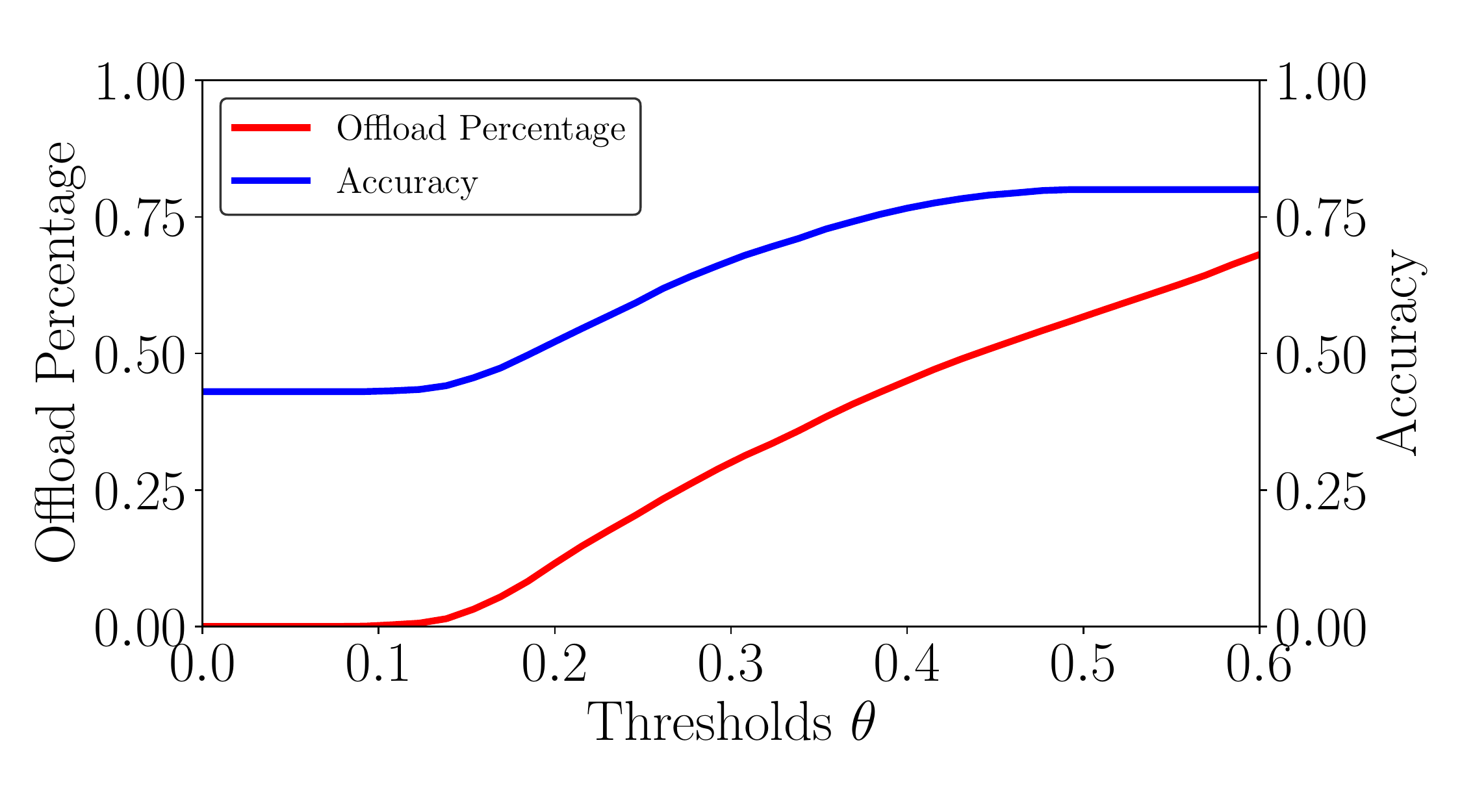}}
	\,
	\subfigure[Accuracy vs. calibrated confidence score.]{\includegraphics[width=0.48\linewidth]{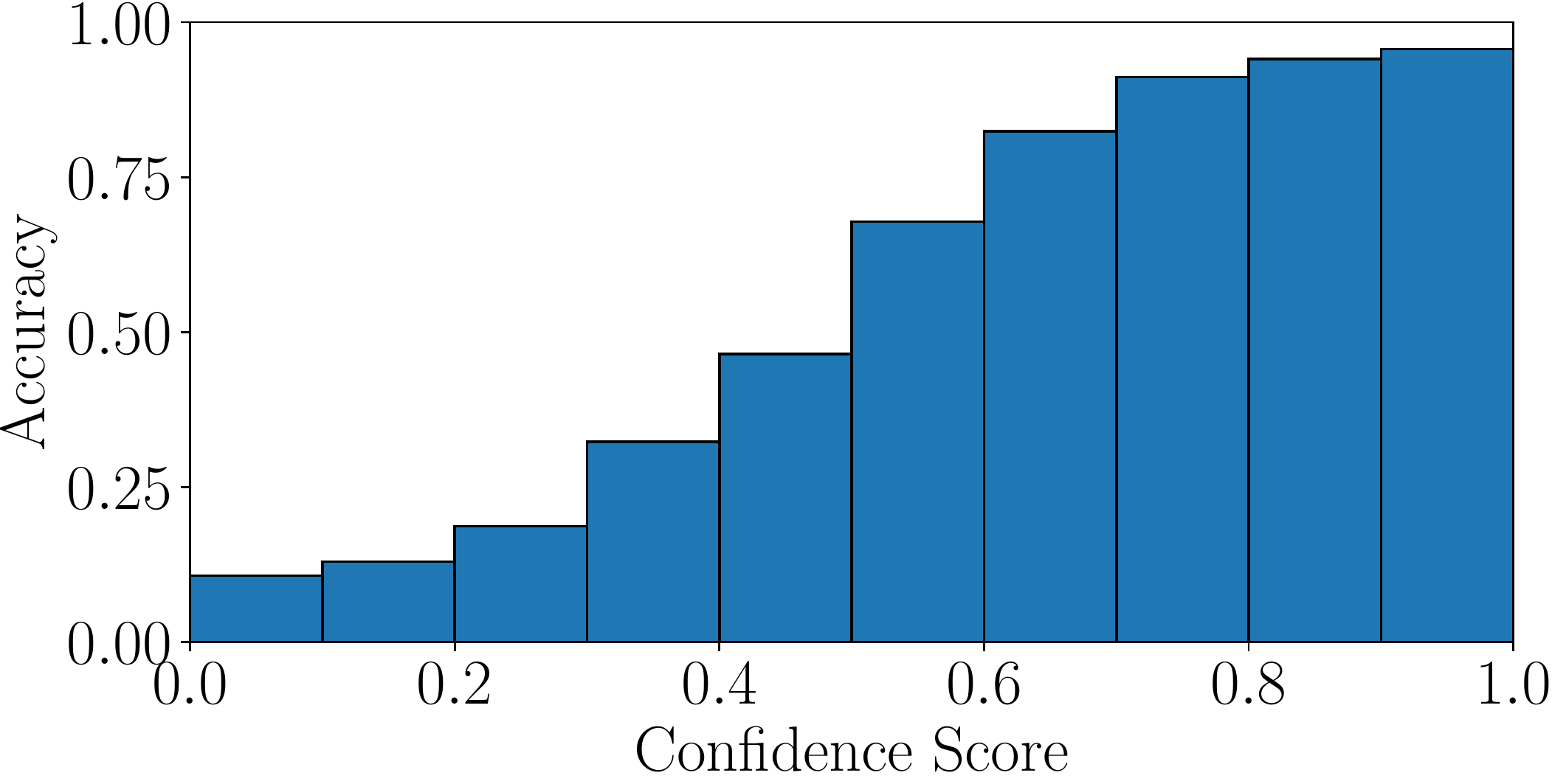}}
	\vspace{-1em}
	\caption{The performance of Platt Calibration.}
	\label{fig:cal-conf-score}
	\vspace{-2em}
\end{figure}

\begin{figure}[t]
	\centering
	\includegraphics[width=0.85\linewidth]{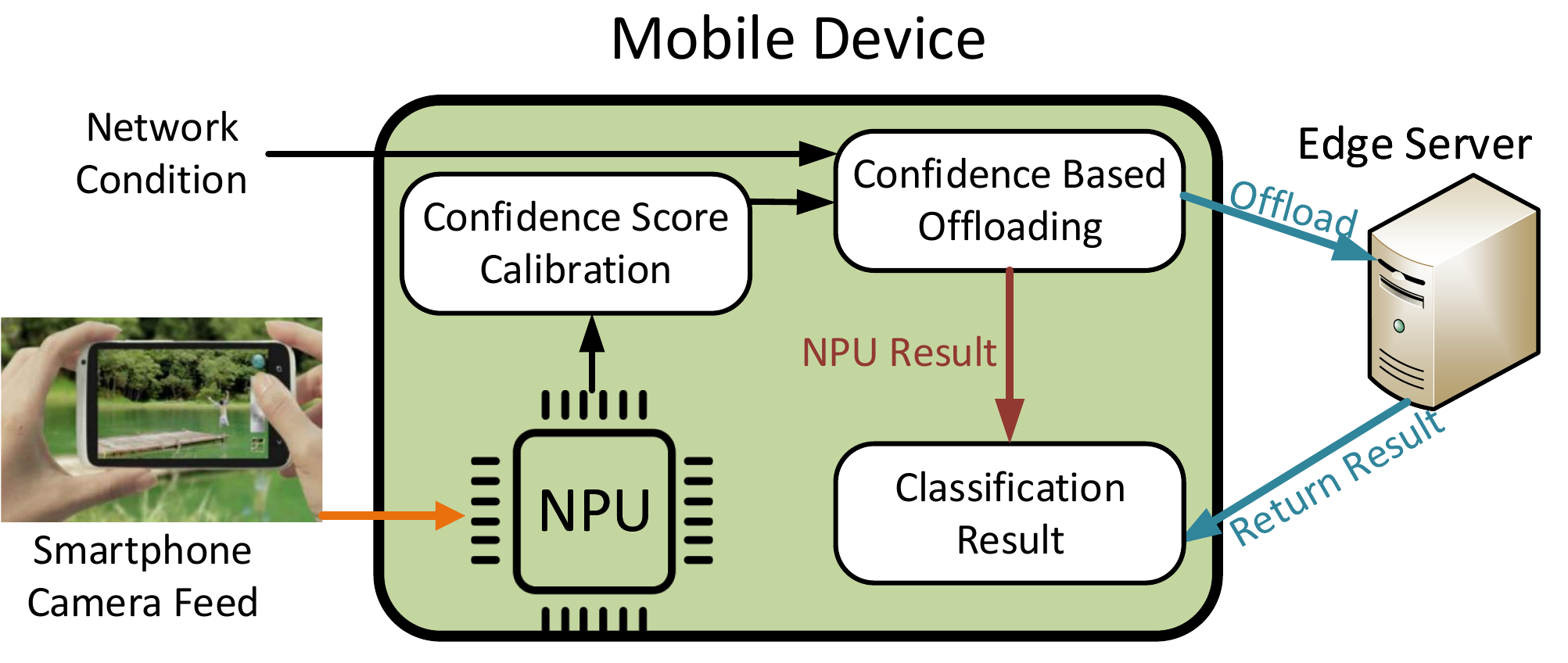}
	\vspace{-0.5em}
	\caption{The CBO framework.}
	\label{fig:overview} 
	\vspace{-1em}
\end{figure}

{\bf The effectiveness of confidence calibration}
We perform experiments to evaluate the performance of the calibration technique using the same setting as that in the last subsection except that the confidence score is replaced with the calibrated confidence score based on Platt Calibration.

As shown in Figure \ref{fig:cal-conf-score} (a),  
the accuracy is increased as the number of offloaded frames increases. 
Compared to the results shown in Figure \ref{fig:org-conf-score}, calibrated confidence score is more effective. 
For example, when the accuracy requirement is 0.7, the threshold has to be above 0.35 at which the offloaded traffic is 30\%.
However, as shown in Figure \ref{fig:org-conf-score}, at least 70\% of the frames are offloaded in order to satisfy the same accuracy requirement.

With the same setting, Figure \ref{fig:cal-conf-score} (b) explains why calibrated confidence score is more effective.
As the confidence score increases from 0 to 1, the accuracy increases from 0.11 to 1.
The accuracy variation is much larger than that shown in Figure \ref{fig:conf-acc-uncal}.
Thus, it is much easier to determine which frames should be offloaded.

\section{The CBO Framework} \label{sec:max-acc}


\subsection{Overview}

Figure \ref{fig:overview} shows our CBO framework for video analytics. 
The video frames are first processed locally on NPU. 
Based on the calibrated confidence score, our framework determines which frames should be offloaded.
To provide real time video analytics, the processing of each video frame should be completed within a time constraint. 
Then, for some offloaded frames, the resolution may be reduced to save bandwidth and delay, at the cost of accuracy.


In the figure, CBO is based on the idea presented in Figure \ref{fig:offload-conf-score}. 
However, the confidence threshold $\theta$ is not fixed; it is adaptively adjusted based on the network condition, the confidence score and the selected frame resolution. 
That is, based on the accuracy and processing time requirement, we study the CBO problem which adaptively selects $\theta$ and the frame resolution, to maximize accuracy under some time constraint.
In the following, we first formulate the problem and then propose an adaptive solution which determines which frames to offload at what resolution based on the confidence score and the network condition.  

\subsection{Problem Formulation}

For each frame $I_i$ ($1 \leq i \leq n$), it is first processed by the DNN on NPU.
Let $p_i$ denote the calibrated confidence score and let $A^{npu}_{p_i}$ denote the accuracy of running the DNN to process the frame on NPU.
Assume the video frame rate is $f$, the time interval between two consecutive video frames is $\gamma = \frac{1}{f}$.
Since NPU is very fast, the local processing time for each frame is shorter than $\gamma$ and it is not the bottleneck.
For the $i^{th}$ video frame, assume its arrival time is $i \gamma$, our system ensures that it is processed before time $T + i \gamma$, where $T$ is the time constraint.

\begin{table}[t]
	\centering
	\begin{tabular}{|c | l |} 
		\hline
		Notation & Description \\ \hline
		$I_i$ & The $i$th frame  \\ \hline
		$A^{o}_r$ & The accuracy of the DNN with input images \\
		& in resolution $r$ \\ \hline
		$S(I_i, r)$ & The data size of the frame $I_i$ in resolution $r$ \\ \hline
		$T^o$ & The processing time on the server \\ \hline
		$B$   & Upload bandwidth (data rate) \\ \hline
		$f$   & Frame rate (fps) \\ \hline
		$T$   & The time constraint for each frame \\ \hline
		$n$	  & The number of video frames that needs to be processed \\ \hline
		$\theta$ & The confidence score threshold for offloading. \\ \hline
	\end{tabular}
	\vspace{-0.5em}
	\caption{Notation}
	\label{table:notations}
	\vspace{-3em}
\end{table}

Based on the calibrated confidence score, a frame may be offloaded to the server for further processing to improve accuracy.
If $p_i$ is higher than the threshold $\theta$, the classification result is returned. Otherwise, 
the classification result is considered to be incorrect and $I_i$ is offloaded to the server in the original resolution or reduced to resolution $r$ before being offloaded.
Let $B$ denote the upload bandwidth and let $L$ denote the network latency between the server and the mobile device.
Then, it takes $\frac{S(I_i, r)}{B} + T^{o} + L$ to transmit the $i^{th}$ frame in resolution $r$ and receive the result from the server.
In this way, the transmission time can be reduced by resizing the frame to a lower resolution, at a cost of lower accuracy.

The notations used in the problem formulation are listed in Table \ref{table:notations}.
The CBO problem can be formulated as an integer programming in the following way.

\begingroup
\allowdisplaybreaks
\setlength\abovedisplayskip{-5pt}
\setlength\belowdisplayskip{0pt}
\small
\begin{align}
	\max \quad & \frac{1}{n} \displaystyle\sum^{n}_{i=1} (A^{npu}_{p_i} (1-X_i) + \sum_r A^o_r Y^r_i X_i) \\
	\textrm{s.t.} \quad 
	& D(k) \leq (i - k) * \gamma + T, \, \forall i, k \\
	& \theta - p_i < X_i, \, \forall i \\
	& p_i - \theta \leq 1 - X_i, \, \forall i \\
	& \sum_r Y^r_i = 1, \forall i \\
	& X_i, Y^r_i \in \{0, 1\}, \, \forall i 
\end{align}
\endgroup

{\noindent}
Where $D(k) = \sum_r \sum_{k\leq i} \frac{S(k, r)Y^{r}_k X_k}{B}  + T^{o} X_i + L$ is the offloading time for the frames that arrive between $I_k$ and $I_i$.
$X_i$ is a variable to show whether the frame is needed to be offloaded and $Y^r_i$ is a variable to show which resolution the frame is resized to before offloading.
If $X_i=0$, the frame $I_i$ is only processed locally.
If $X_i=1$, the frame $I_i$ is offloaded to the server for further processing.
If $Y^r_i = 1$, the frame $I_i$ is resized to resolution $r$ before offloading.

Objective (1) is to maximize the accuracy of the processed frames.
Constraint (2) specifies that the classification results should be returned within the time constraint, and constraint (3) specifies that the frames with confidence score lower than or equal to threshold $\theta$ should be offloaded.
Constraint (4) specifies that frames with confidence score higher than $\theta$ should not be offloaded.

\subsection{Finding the Optimal Solution}

In this subsection, we build a solution graph based on which we can find the optimal solution for the CBO problem.
As shown in Figure \ref{fig:solution-graph}, nodes at different levels represent the schedule options for different frames. 
More specifically, nodes at level $i (1\leq i \leq n)$ represent the schedule options of frame $I_i$. 
There are $(m+1)$ nodes at level $i$, representing $m+1$ scheduling options, where $m$ is the number of frame resolutions. 
For example, at level 1 (i.e., for frame $I_1$), node $V^{npu}_{1}$ represents that the frame is processed on local NPU, 
and $V^{r_m}_{1}$ represents that frame $I_1$ is offloaded to the server in resolution $r_m$. 
Each node is associated with a time window $[i\gamma, i\gamma + T]$ which represents the time constraint for the frame to be processed. 
We also create two dummy nodes: $V_{start}$ and $V_{end}$.
They are the source node and the destination node at level 0 and $n+1$.

\begin{figure}
	\centering
	\includegraphics[width=0.75\linewidth,valign=c]{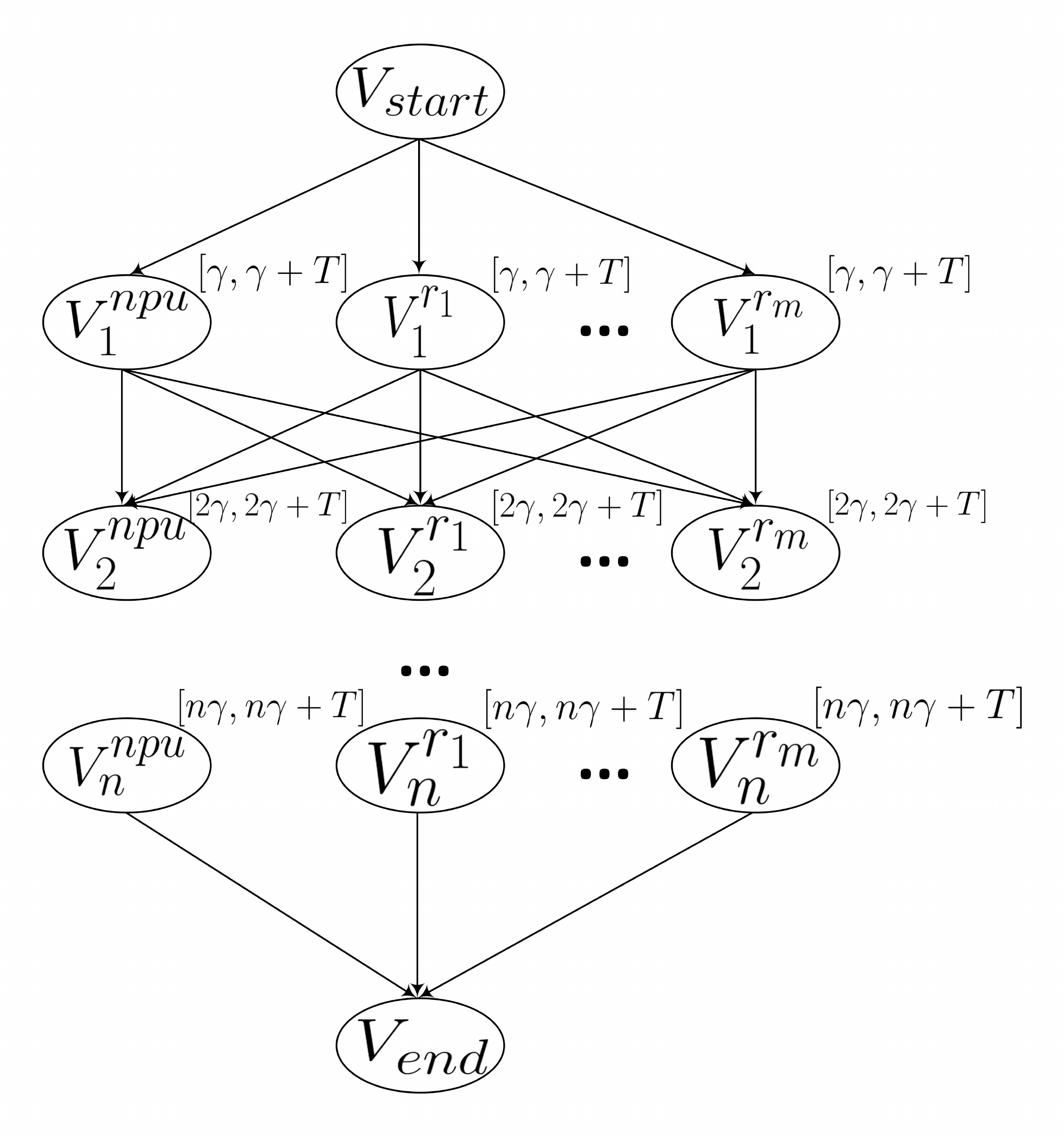}
	\vspace{-1em}
	\caption{The solution graph.}
	\label{fig:solution-graph}
	\vspace{-2em}
\end{figure}

For edges, we add a link from each node at level $i$ to all nodes at level $i+1$.
An edge $(V^r_i, V^r_{i+1})$ has two attributes, cost and time duration.
The cost is defined as the negative of the accuracy $c(V^r_i, V^r_{i+1}) = -A^o_r$ and $c(V^{npu}_i, V^r_{i+1})=A^{npu}_{p_i}$.
The time duration is the offloading time since local processing time on NPU is very short and it is not the bottleneck. 
The time duration between $V^r_i$ and $V^r_{i+1}$ (for frame $I_i$ with resolution $r$) can be computed as $t(V^r_i, V^r_{i+1}) = \frac{S(i, r)}{B}$.
In this way, our problem is converted to the problem of finding the least cost path from $V_{start}$ to $V_{end}$ while visiting each chosen node within its specific time window.

\begin{theorem}
	The CBO problem is NP-hard.
\end{theorem}

\begin{proof}
	We reduce a well known NP-hard problem, the subset sum problem to our problem.
	In the subset sum problem, there is a set $U$ which includes $n$ numbers $(a_1, a_2, a_3, \ldots, a_n)$ and the goal is to find a subset of numbers so that its sum is equal to a value $K$.
	
	For an arbitrary instance of the subset sum problem, we can construct an instance of our solution graph as follows.
	For each number $a_i$, two nodes $V^0_i$ and $V^1_i$ are added to the graph at level $i$.
	The time window of these two nodes are set to be $[-\sum_i |a_i|, \sum_i |a_i|]$.
	Specially, the time window for the destination node $V_{end}$ is set to be $[K, K]$.
	
	For edges, we add links from nodes at level $i$ to nodes at level $i+1$.
	More specifically, the cost is set to be 0 and the time duration $t(V^p_i, V^q_{i+1})$ is set to $a_i$ if $p=1$; otherwise, $t(V^p_i, V^q_{i+1})=0$.
	
	
	A solution to our problem must satisfy the requirement that the destination node $V_{end}$ must be visited at the exact time $K$.
	As a result, the sum of the time duration of the selected nodes is equal to $K$.
	Since the time duration of a node is also equivalent to the value of the corresponding number, the sum of the selected numbers is also equal to $K$.
	This completes the reduction and hence the proof.
\end{proof}

Due to the time window constraint, the shortest path cannot be found using the Dijkstra algorithm.
Instead, a dynamic programming algorithm is used to search for the optimal path.

Let $P_j(V^r_i)$ denote the $j^{th}$ feasible path from $V_{start}$ to $V^r_i$ and each path has two attributes $T_j(V^r_i)$ and $C_j(V^r_i)$, which are used to record the time duration and the cost of the path.
Initially, $T(V_{start})$ and $C_(V_{end})$ are set to be 0.
For each path $P_j(V^r_i) = (V_{start}, V^r_1, \ldots, V^r_i)$, its attributes are iteratively computed as follows

\begingroup
\allowdisplaybreaks
\setlength\abovedisplayskip{-13pt}
\setlength\belowdisplayskip{0pt}
\begin{align*}
	& T_j(V^r_k) = \max(T_j(V^r_{k-1}) + t(V^r_{k-1}, V^r_k), k\gamma) \\
	& C_j(V^r_k) = C_j(V^r_{k-1}) + c(V^r_{k-1}, V^r_k) 
\end{align*}
\endgroup

Since the node $V_i$ must be visited during $[i\gamma, i\gamma + T]$, a feasible path should satisfy $i \gamma \leq T_j(V^r_k) \leq i\gamma + T$.
Although a lot of feasible paths can be found in the iterations, the algorithm only considers the most efficient ones.
More specifically, for two paths $P_1(V^r_i)$ and $P_2(V^r_i)$, if $T_1(V^r_i) < T_2(V^r_i)$ and $C_1(V^r_i) < C_2(V^r_i)$,  $P_1(V^r_i)$ is more efficient than $P_2(V^r_i)$ and $P_2(V^r_i)$ will not be considered in future iterations.
The optimal path is $P(V_{end})$ which has the minimum cost $\min C(V_{end})$.

Since there are at most $T$ different efficient paths from $V_{start}$ to $V^r_i$, there are at most $mT$ paths at level $i (1 \leq i \leq n)$.
Therefore, the time complexity of the optimal algorithm is $O(nm^2T)$.

\subsection{The CBO Algorithm}


The optimal solution can maximize the accuracy within the time constraint.
However, it is not practical since it requires the complete knowledge of all frames, such as the frame sizes and the confidence score of running DNN to process the frame.
In this subsection, we remove this assumption and propose an adaptive solution, called CBO algorithm. 

Since the frames with lower confidence scores are classified with lower accuracy on NPU, 
they should be offloaded to increase the accuracy as long as there is available bandwidth.
However, due to bandwidth limitation, some frames cannot be offloaded and have to rely on local NPU for classification. 

We use the following dynamic programming algorithm to determine which frames should be offloaded with what resolution. 
Suppose $k$ frames have been processed locally.
For each frame $I_i$ ($1 \leq i \leq k$), its arrival time is $t^{arr}_i$ and the confidence score is $p_i$.
The frames are sorted in the descending order of the confidence scores, which means $p_i > p_j$ if $i < j$.
In our algorithm, a list $l_j$ ($j \in [0, k]$) is used to find the schedule decision for maximizing the accuracy.
Each element in the list $l_j$ is a pair $(t, A)$, where $A$ is the accuracy improvement which can be achieved by offloading the first $j$ frames within time $t$.
If  $I_i$ is offloaded in resolution $r$, the accuracy improvement can be computed as $A=A^o_r - A^{npu}_{p_i}$.
Initially, $l_0 = {(0, 0)}$.
To add pairs to the list $l_j$, we consider the following two cases.

{\em No offloading: } In this case, the $j^{th}$ frame will not be offloaded to the server for further processing.
All pairs in $l_{j-1}$ will be added to $l_j$.

{\em Offloading: } In this case, the $j^{th}$ frame will be offloaded in resolution $r$.
It takes $\frac{S(I_j, r)}{B}$ to transmit the frame and the accuracy improvement is $A^o_r - A^{npu}_{p_j}$.
For each pair $(t, A) \in l_{j-1}$, a new pair $(max(t, t^{arr}_j) + \frac{S(I_j, r)}{B}, A + A^o_r - A^{npu}_{p_j})$ is added to $l_j$.
Notice that the frames should be processed within the time constraint.
Therefore, $max(t, t^{arr}_j) + \frac{S(I_j, r)}{B} + T^o + L \leq T + t^{arr}_j$ must be satisfied for all new pairs.

To improve the efficiency of our algorithm, only the most efficient pairs in $l_j$ are kept.
More specifically, a pair $(t', A')$ is said to dominate another pair $(t, A)$ if and only if $t' \leq t$ and $A' \geq A$.
The pair $(t', A')$ is more efficient than the pair $(t, A)$ and all the dominated pair will be removed from $l_j$.

With the list of $l_k$, we can find the confidence threshold $\theta$ and offloaded frame resolution $r^o$.
After the first frame is offloaded, the algorithm will be run again for the frames that have been processed locally.
The CBO algorithm is summarized in Algorithm \ref{alg:cbo}.
In Lines 1-10, dynamic algorithm is applied to maximize the accuracy, and the schedule decision is determines for frame $I_1$ in Lines 11-20.
The running time of the algorithm is $O(k^2 * m)$.

\setlength{\textfloatsep}{-3pt}
\begin{algorithm}[t]
	\SetAlgoLined
	\KwResult{Confidence threshold $\theta$, frame resolution $r^o$}
	\SetKw{Break}{break}
	\DontPrintSemicolon
	
	$l_0 \leftarrow \{(0, 0)\}$ \\
	
	\For{$j \leftarrow 1$ to $k$}{
		\For{each $(t, A) \in l_{j-1}$}{
			Add $(t, A)$ to $l_j$ \\
			\For{each possible resolution r}{
				$t' \leftarrow max(t, t^{arr}_j) + \frac{S(j, r)}{B}$ \\
				\If{$t' + T^o + L \leq T + t^{arr}_j$}{
					$A' \leftarrow A^o_r - A^{npu}_{p_j}$ \\
					Add $(t', A + A')$ to $l_j$ \\
				}
			}
		}
		Remove the dominated pairs from $l_j$ \\
	}
	
	$(t', A') \leftarrow \arg \max_{(t, A) \in l_k} A$ \\
	
	\For{$j$ from $k - 1$ to $0$}{
		\For{each pair $(t, A)$ in $l_j$} {
			\For{each possible resolution $r$}{
				\If{$t + \frac{S(I_j, r)}{B} = t'$ and $A + A^o_r - A^{npu}_{p_j} = A'$}{
					$A' \leftarrow A$, $t' \leftarrow t$ \\
					$\theta \leftarrow p_j$, $r^o \leftarrow r$ \\
				}
			}
		}
	}
	
	\Return $r^o$, $\theta$ \\
	\caption{CBO algorithm}
	\label{alg:cbo}
\end{algorithm}


\section{Performance Evaluations} \label{sec:evaluation}

In this section, we evaluate the performance of the CBO algorithm and compare it with other approaches.

\subsection{Evaluation Setup}

The evaluations are performed on HUAWEI mate 10 pro, which is equipped with 6 GB memory, octa-core CPU ($4 \times 2.4$ GHz and $4 \times 1.8$ GHz) and a NPU.
HUAWEI has published the HUAWEI DDK \cite{hiai} toolset for developers to run the DNN on NPU.
The pre-trained DNNs must be optimized before they can be run on NPU since NPU has a different architecture from CPU.
The HUAWEI DDK includes toolsets to perform such optimizations for DNNs.
It also includes the APIs to run the DNNs, and a few Java Native Interface (JNI) functions are provided to use the APIs on Android.
Since these JNI functions cannot extract the confidence scores,
we add JNI functions to use the confidence score calibration model which is trained on a powerful desktop.

In the experiment, AlexNet is deployed on the mobile device and ResNet-152 is deployed on the server.
These DNNs are used for object recognition,
and they are very popular in computer vision community and have been fine-tuned for many problem.
Moreover, AlexNet has a simple model structure and it can be executed efficiently on NPU to provide real-time video analytic.
In contrast, ResNet-152 is more complex and it can achieve higher accuracy than AlexNet at the cost of more computational power.

To measure the performance of our algorithm, we use a subset of videos from the FCVID dataset, which includes many real-world videos.
These videos have been used for training models related to object classification and activity recognition.
In our experiment, we focus on object classification, and thus activity recognition clips are not used.
Since the dataset is very large, about 1.9 TB, we randomly select 40 videos from the dataset and filter out the noisy data. 

In the experiment, the frames are offloaded in the lossless PNG format.
The server is a desktop with AMD Ryzen 7 1700 CPU, GeForce GTX1070 Ti graphics card and 16 GB RAM.
We have installed the Caffe \cite{jia-mm14} framework to run the DNNs on GPU.

\begin{figure}[t]
	\centering
	\includegraphics[width=0.8\linewidth,valign=c]{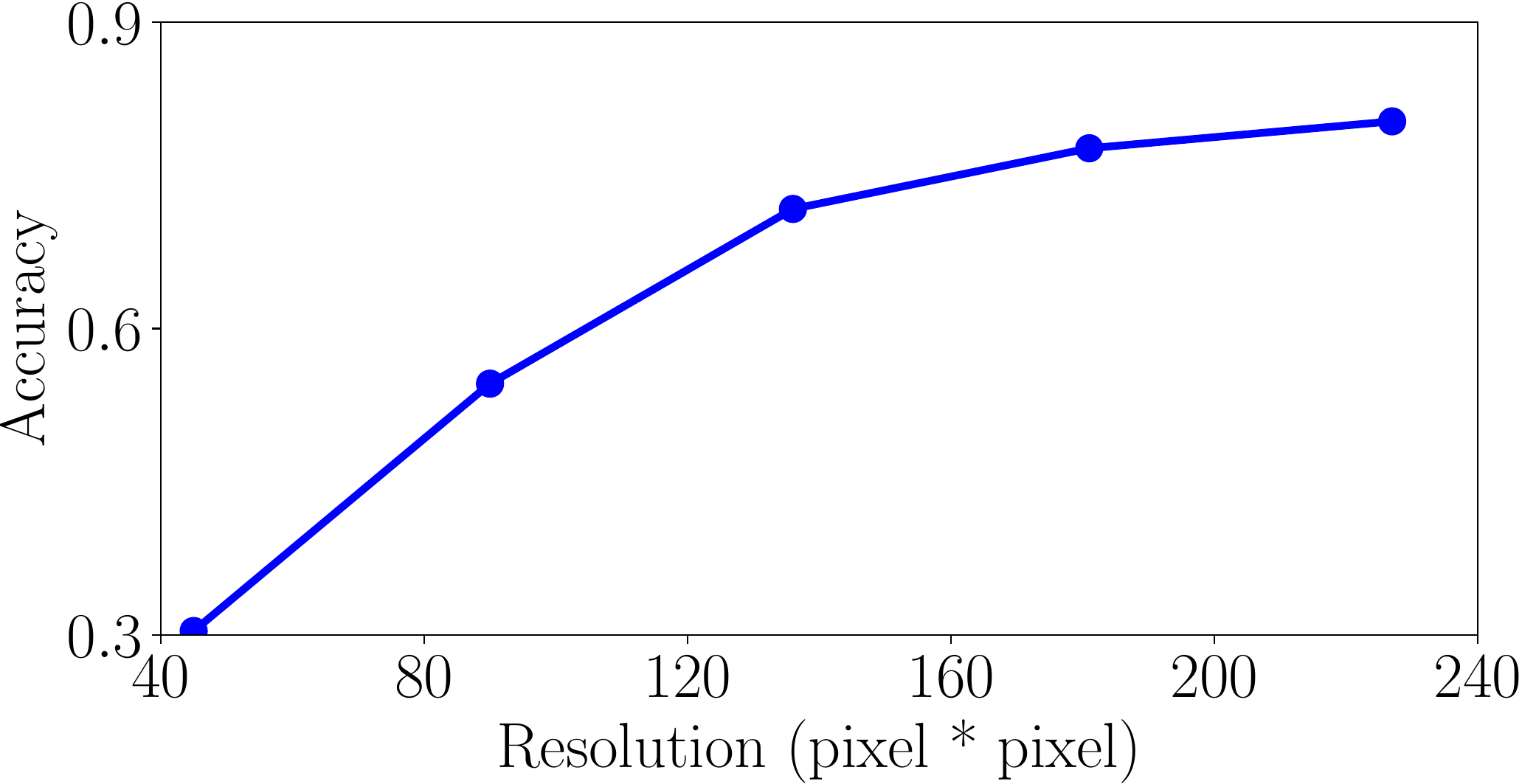}
	\vspace{-0.5em}
	\caption{\small Accuracy vs. Resolution.}
	\label{fig:exp-resolution}
\end{figure}

\begin{figure*}[h]
	\vspace{-1em}
	\centering
	\begin{minipage}[t]{0.32\textwidth}
		\centering
		\includegraphics[width=\linewidth]{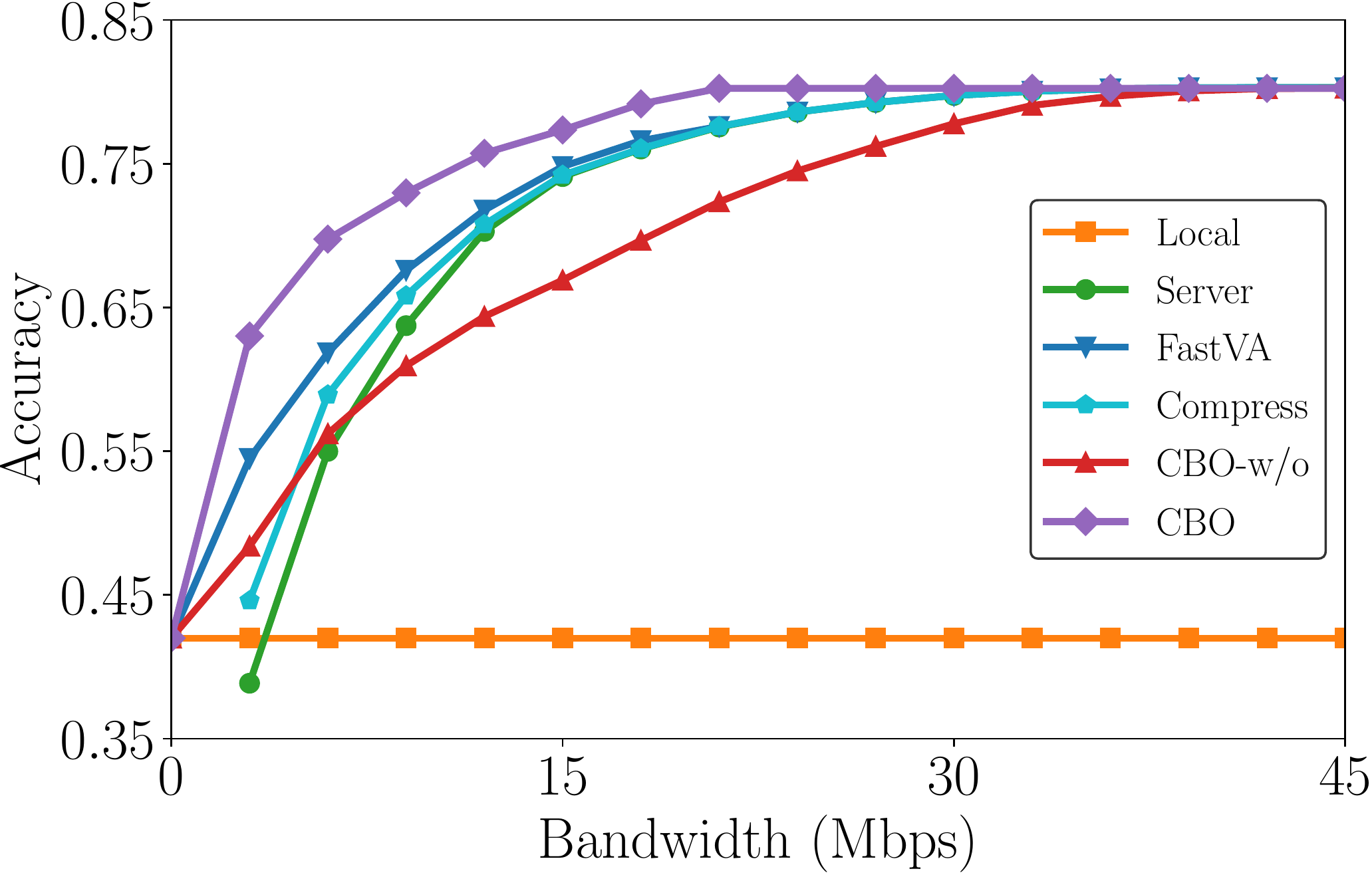}
		\vspace{-1.5em}
		\caption{The performance of different approaches under different network conditions.} \label{fig:exp-bandwidth}
	\end{minipage}
	\,
	\begin{minipage}[t]{0.32\textwidth}
		\centering
		\includegraphics[width=\linewidth]{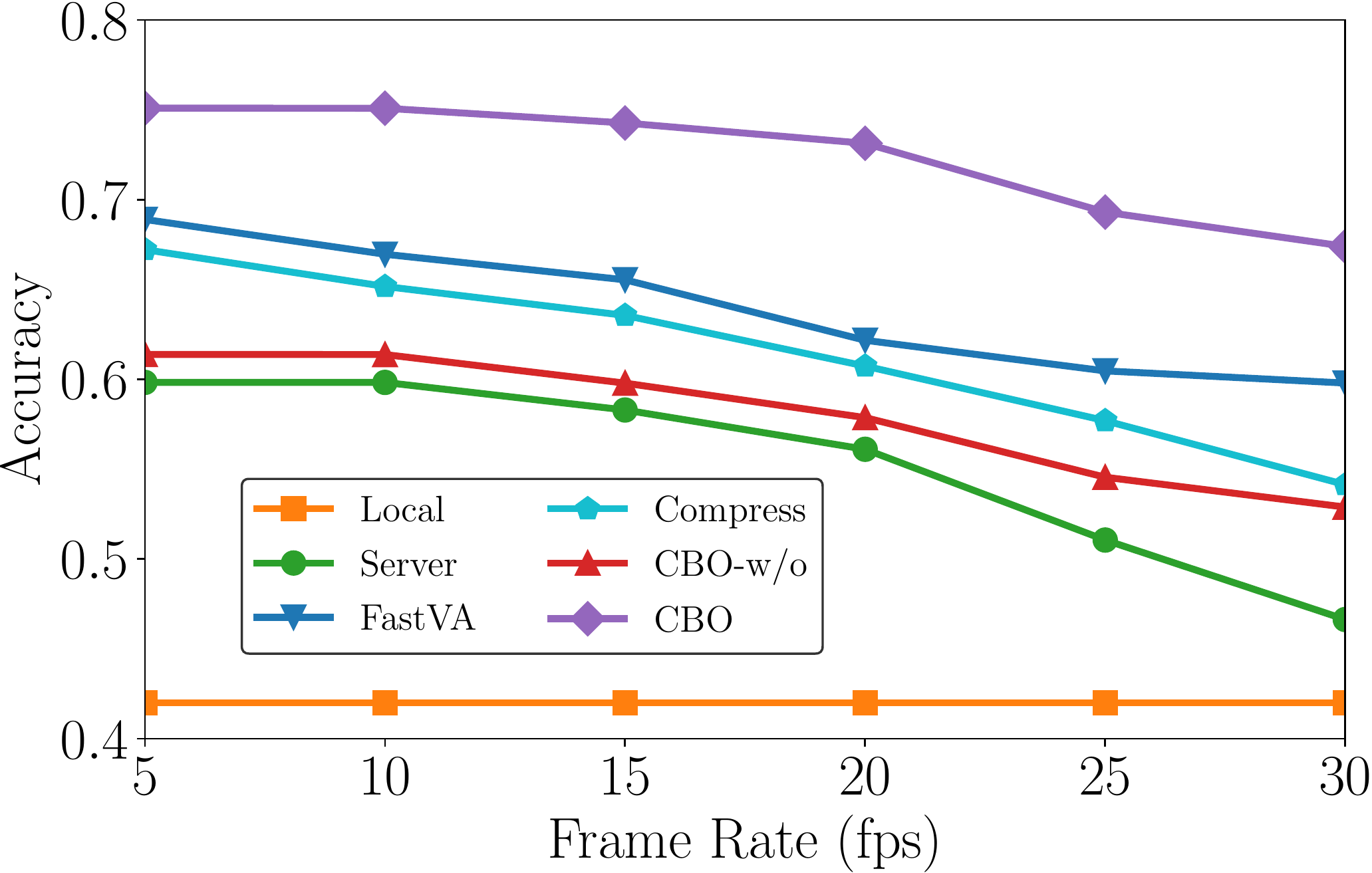}
		\vspace{-1.5em}
		\caption{The performance of different approaches under different frame rates.}\label{fig:exp-fps}
	\end{minipage}
	\,
	\begin{minipage}[t]{0.32\textwidth}
		\centering
		\includegraphics[width=\linewidth]{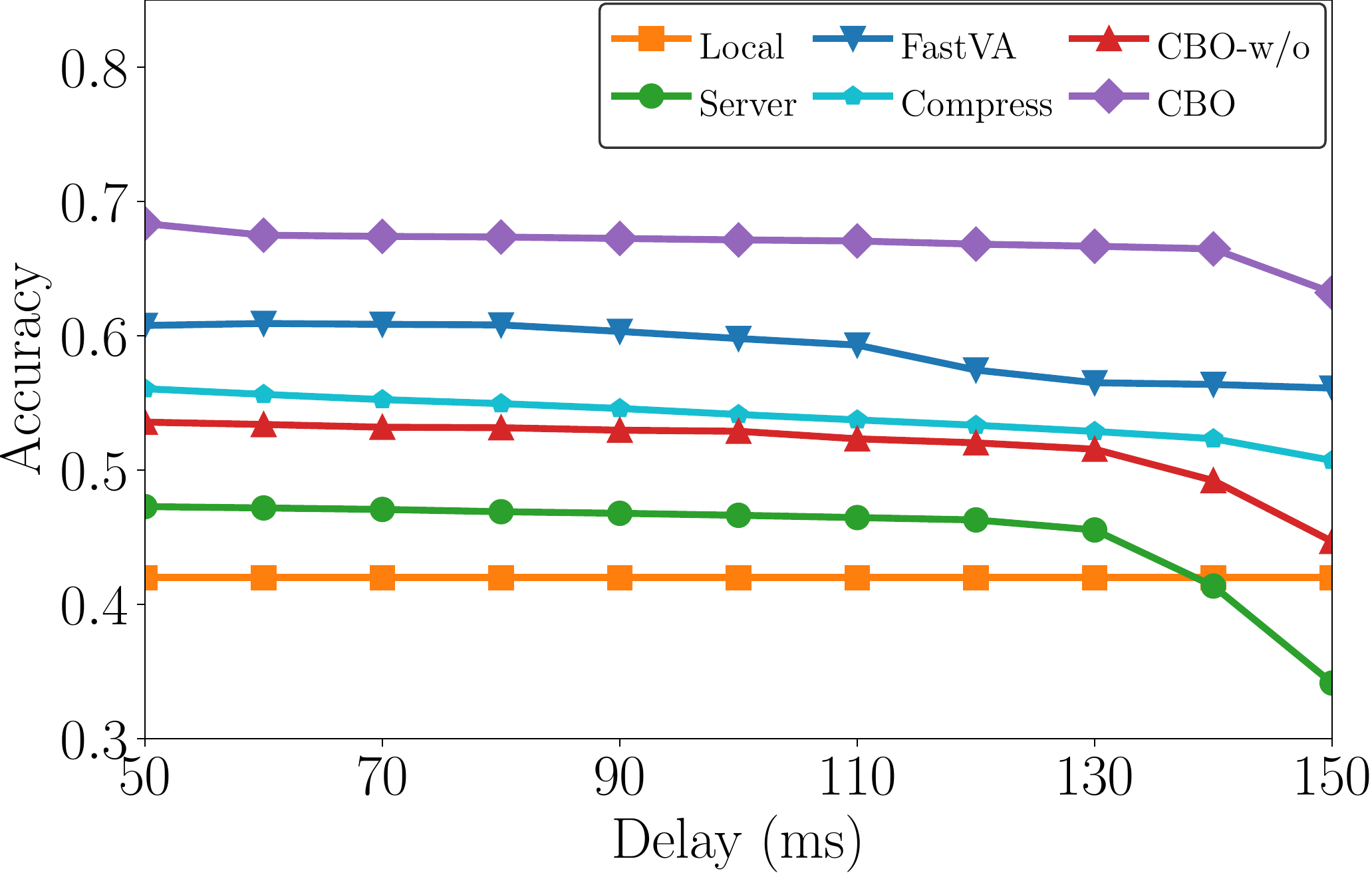}
		\vspace{-1.5em}
		\caption{The performance of different approaches under different network latency.}\label{fig:exp-delay}
	\end{minipage}
	\vspace{-1.5em}
\end{figure*}

We evaluate the proposed algorithms with different frame rates.
Most videos in the dataset use 30 fps, and thus we have to change their frame rate by decoding/encoding.
For ResNet-152, the maximum resolution of the input image is 224x224 pixels.
This resolution can be downsized for some offloading images, and we consider 5 different resolutions: 45x45, 90x90, 134x134,  179x179 and 224x224 pixels.
The tradeoff between accuracy and frame resolution is shown in Figure \ref{fig:exp-resolution}.

\begin{table}[h]
	\centering
	\begin{tabular}{|c|c|}
		\hline
		& Running Time (ms) \\ \hline
		AlexNet on NPU               & 20                \\ \hline
		ResNet-152 on server         & 37                \\ \hline
		Confidence Score Calibration & 8                 \\ \hline
	\end{tabular}
	\vspace{-0.5em}
	\caption{The running time of different models for each frame.}
	\label{table:time-measurement}
\end{table}
\vspace{-1em}

The time constraint for each frame is set to be 200 ms in all experiments.
The running time of our CBO algorithm is less than 1 ms on smartphone and it is negligible compared to the time constraint (100 ms level).
The running time of different DNNs are shown in Table \ref{table:time-measurement}.
The data shown in Figure \ref{fig:exp-resolution} and Table \ref{table:time-measurement} will be used for making scheduling decisions in Algorithm \ref{alg:cbo}.
We compare CBO with the following approaches.

\begin{itemize}
	\item \textbf{Local: } 
	All frames are processed locally on NPU.
	After a frame has been processed on NPU, the classification result is returned. 
	
	\item \textbf{Server: } 
	All frames are offloaded to the server for processing.
	Each frame is resized to a resolution due to bandwidth limitation so that it can be offloaded before the next frame arrives. 
	
	\item \textbf{FastVA: }
	This is an implementation of the FastVA framework \cite{tan-infocom2020} which maximizes accuracy under some time constraint.
	Based on the optimization, it determines which frames should be processed locally on NPU and which frames should be offloaded to the edge server.
	
	\item \textbf{Compress: }
	This method is similar to FastVA, except that it runs a compressed DNN to process the frames locally on CPU.
	We leverage Tensorflow which is a well-known deep learning framework to perform pruning and dynamic range quantization on the original DNN. 
	Specifically, the sparsity of the model is set to 0.9 in the pruning step.
	After compression, the model size is reduced by 95\% and 50\% of the processing time is saved on CPU, but still much slower than NPU.
	
	
	\item \textbf{CBO without Calibration (CBO-w/o): } 
	This algorithm is the same as CBO, except that the confidence score is not calibrated.
	
	\item \textbf{Optimal: }
	This shows the performance upper bound for all algorithms.
	It tries all possible combinations and chooses the schedule that maximizes the accuracy.
	Note that this method cannot be used for processing videos in real time since it takes too much time to search all possible schedules. 	We can only find the optimal solution offline by replaying the data trace.
\end{itemize}

\subsection{Evaluation Results}

The performance of the algorithm depends on several factors, the bandwidth, the network latency, and the video frame rate.

In Figure \ref{fig:exp-bandwidth}, we compare CBO with Local, Server, FastVA, Compress and CBO-w/o under different network conditions.
In the evaluation, the network latency is set to be 100 ms and the frame rate is set to be 30 fps.
Since no frame is offloaded to the server in the Local approach, its performance does not change under different network conditions.
Local, FastVA, CBO and CBO-w/o have the same accuracy when the bandwidth is 0, since no frame can be offloaded to the server.
When the bandwidth is lower than 3 Mbps, the Local approach achieves higher accuracy than the Server approach.
This is because the frames have to be offloaded in an extremely low resolution.
Even with advanced DNN running at the server, the accuracy is still pretty low with low resolution frames.
Compared to the Compress approach, FastVA achieves higher accuracy when the network bandwidth is lower than 15 Mbps.
This is because the running time of compressed DNN on CPU is long and Compress offloads most frames in low resolution.
Note that the processing time of compressed DNN (50\%) is not reduced as much as the model size (95\%).
This is because the reduction of the model size is mainly due to the removal of the redundant parameters in the fully connected layers. After compression, there are still many convolutional layers which are computationally intensive, and it takes a large amount of time to run these layers.
In contrast, FastVA avoids offloading frames in low resolution by processing video frames on NPU.
As can be seen from the figure, CBO outperforms FastVA, since CBO knows which frames are classified incorrectly based on the calibrated confidence scores and it can improve the accuracy effectively by offloading the frames with low confidence scores.

From Figure \ref{fig:exp-bandwidth}, we can also see that CBO-w/o underperforms CBO and FastVA because the uncalibrated confidence scores cannot accurately estimate the correctness of the classification result. 
With uncalibrated confidence score, CBO-w/o may offload frames which have been classified correctly on NPU, wasting bandwidth resources, and it may return misclassified results and reduce the accuracy. 
Compared to the Server approach, CBO-w/o achieves higher accuracy when the network bandwidth is low.
This is because CBO-w/o avoids offloading frames in very low resolution by returning the classification result with high confidence score.
When the bandwidth is high, the Server approach outperforms CBO-w/o since it can offload frames with higher resolution and then increase the accuracy.
As the bandwidth increases above 36Mbps, the difference among CBO, CBO-w/o, FastVA, Compress and Server becomes smaller since most frames can be offloaded to the server in higher resolution to achieve higher accuracy.

In Figure \ref{fig:exp-fps}, we evaluate the impact of frame rate on accuracy for different approaches.
We set the uplink network bandwidth to be 5 Mbps and set the network latency to be 100 ms.
As shown in the figure, CBO significantly outperforms CBO-w/o, Server, FastVA, Compress and Local. 
In general, the accuracy of all approaches drops when the frame rate increases.
The Server approach suffers a 15\% accuracy drop when the frame rate increases from 5 fps to 30 fps.
This is because most frames have to be resized to low resolution when the frame rate increases to 30 fps. 
In contrast, only 6\% accuracy drop is observed in CBO, since it only offloads the frames with low confidence scores. With the same amount of bandwidth, CBO can offload more frames with high resolution compared to other approaches.

In Figure \ref{fig:exp-delay}, we evaluate the impact of network latency on accuracy.
We set the uplink network bandwidth to be 5 Mbps and set the frame rate to be 30 fps.
As shown in the figure, CBO significantly outperforms CBO-w/o, Server, FastVA, Compress and Local.
Since the Local approach does not offload any frames, its performance remains the same.
As the latency increases, less frames can be offloaded to the server for processing due to the delay constraint requirement.
Therefore, the performance of Server, Compress, FastVA, CBO and CBO-w/o degrades as the latency increases.
Compared to the Server approach, the accuracy drop in CBO is much smaller.
This is because CBO can reduce the chance of offloading frames in low resolution by using NPU based classification results that have high level of confidence.

\begin{figure}[t]
	\centering
	\footnotesize 
	\subfigure[The performance of Optimal]{\includegraphics[width=0.48\linewidth]{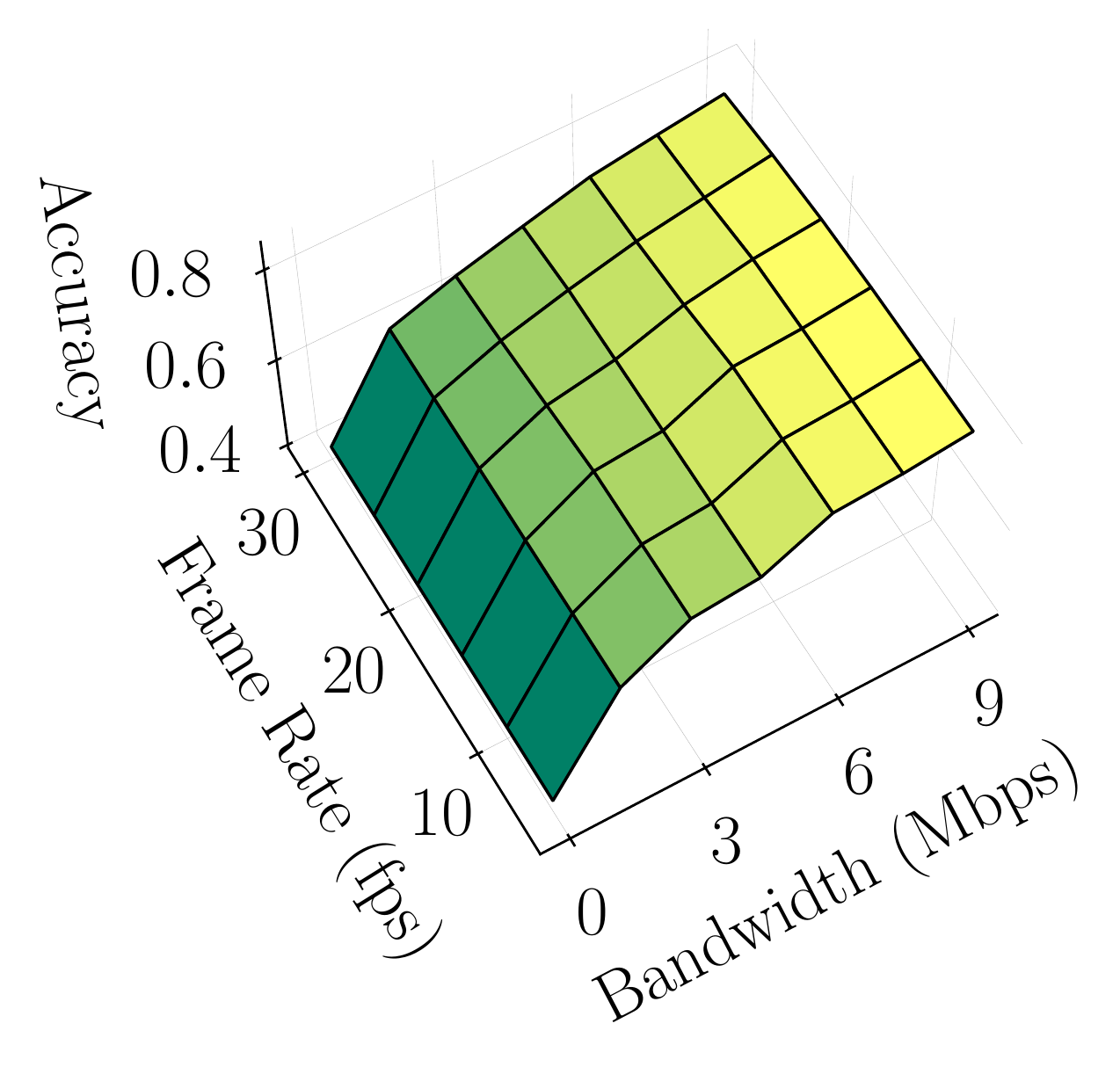}}
	\,
	\subfigure[Difference between Optimal and CBO]{\includegraphics[width=0.48\linewidth]{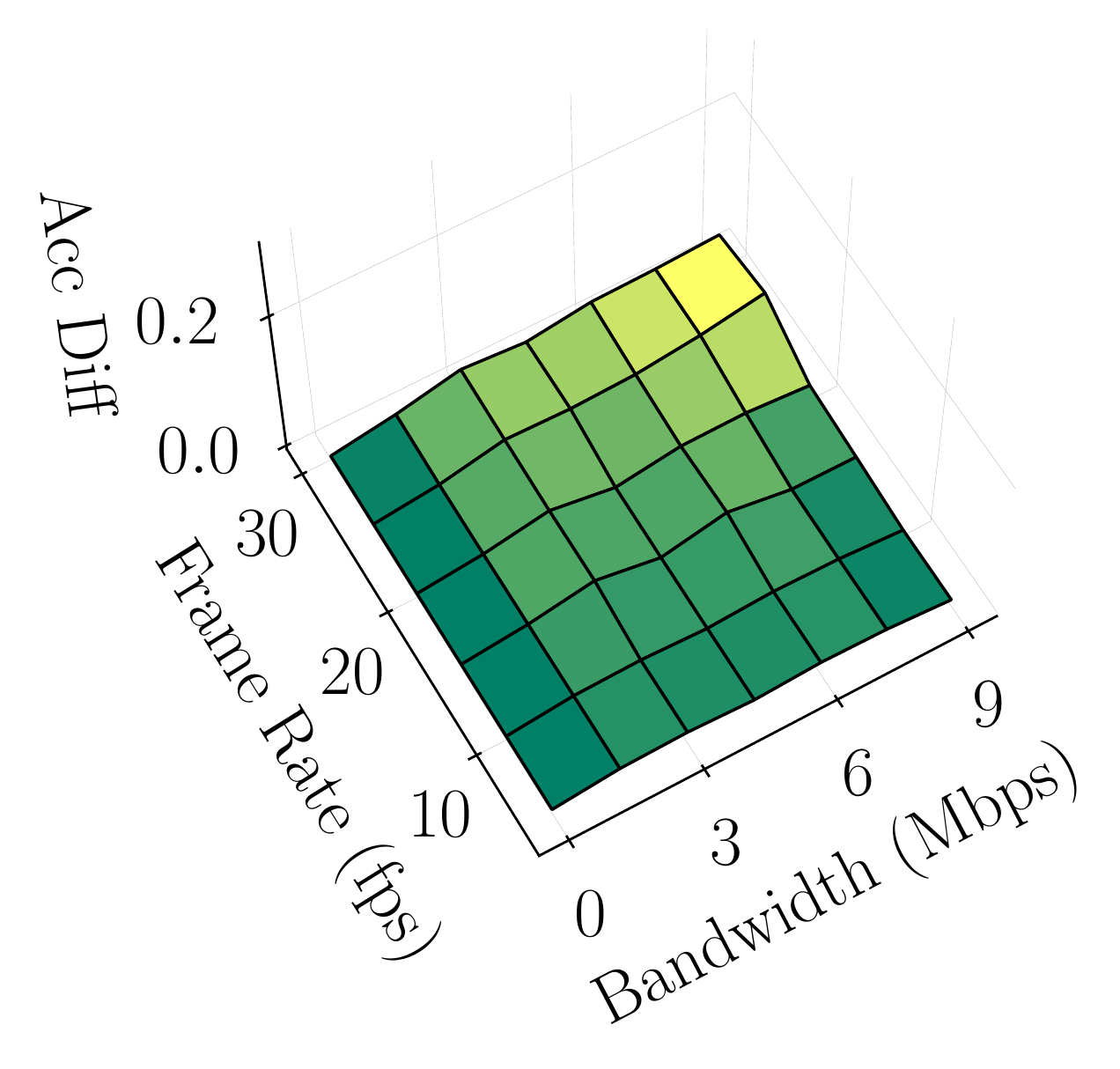}}
	\vspace{-0.9em}
	\caption{Comparison between CBO and Optimal.}
	\label{fig:3d-optimal}
\end{figure}

In Figure \ref{fig:3d-optimal}, we compare CBO with Optimal under various frame rate and network conditions.
As shown in Figure \ref{fig:3d-optimal}(a), the accuracy of Optimal increases when the network bandwidth increases, because the mobile device can upload more frames with higher resolution. 
As frame rate requirement increases, less frames can be offloaded and Optimal has to offload frames with lower resolution.
Therefore, the accuracy becomes lower.

In Figure \ref{fig:3d-optimal}(b), we show the accuracy difference between Optimal and CBO.
The accuracy difference is computed using the accuracy of the Optimal approach minus that of CBO. 
As can be seen from the figure, the difference is almost zero in most cases, which indicates that CBO is close to Optimal.

\section{Related Work} \label{sec:related-work}

In recent years, researchers have made tremendous progress in applying DNNs for various classification problems \cite{parkhi-bmvc15, lapuschkin-cvpr16, he-cvpr16, barbu-nips19}.
However, these DNNs are designed for machines with powerful CPU and GPU, and it is difficult to run them on mobile devices due to the resource limitations.
To address this issue, researchers leveraged model compression technique to reduce the resource demand of DNNs.
For example, in \cite{luo-iccv17, yu-cvpr18}, the authors proposed techniques to remove the redundant parameters and operations from the neural network to reduce the model size and processing time.
Although the efficiency can be improved through these model compression techniques, the accuracy also drops.

Computation offloading represents another kind of solution for enabling video analytics on mobile devices.
Some general offloading frameworks \cite{cuervo-mobisys10, Geng-ICNP2015, Geng-infocom18} have been proposed to optimize energy and reduce the computation time for mobile applications.
However, these frameworks have limitations when applied to deep learning based video analytics since a large amount of data has to be offloaded to the server.
To address this issue, researchers propose offloading framework for deep learning applications \cite{han-mobisys16, tan-tmc21, chen-sensys18,  ran-infocom18, tan-secon21, liu-mobicom2019}.
There have been some studies on confidence based offloading.
For example, in \cite{teerapittayanon-icdcs17} and \cite{wang-infocom2020}, 
confidence score is leveraged to reduce the processing delay by early exit; i.e., returning the classification results without running all DNN layers.
Different from them, our framework is designed for video analytics on mobile device with NPU, where the goal is to maximize accuracy under some time constraint. Moreover, we propose confidence score calibration technique to improve the performance.

Considerable amount of work has been done on improving the execution efficiency of DNNs on mobile devices through hardware support. 
For example, Tan {\em et al.} \cite{tan-ipsn21} developed model partitioning techniques to schedule some neural network layers on CPU while executing other layers on NPU to achieve better tradeoffs between processing time and accuracy.
Cappuccino \cite{motamedi-esl19} optimized computation by exploiting imprecise computation on the mobile system-on-chip (SoC).
Oskouei \textit{et al.} \cite{oskouei-mm16} developed an Android library called CNNdroid for running DNNs on mobile GPU.
DeepMon \cite{huynh-mobisys17} leveraged GPU for continuous vision analysis on mobile devices.
FastVA \cite{tan-infocom2020} leveraged NPU and offloading technique for video analytics on mobile devices.
Different from FastVA, we explore insights about running DNNs on NPU to improve performance.

\section{Conclusions} \label{sec:conclusion}

In this paper, we proposed a CBO framework for video analytics to address the low accuracy problem of running DNNs on NPU. 
The major challenge is to determine when to return the NPU classification result based on the confidence level of running the DNN, and when to offload the video frames to the server for further processing to increase the accuracy. We found that existing confidence scores were not effective for making offloading decisions, and thus proposed techniques to calibrate the confidence score so that it could accurately reflect the correctness of the classification results on NPU. 
We formulated the CBO problem, where the goal is to 
maximize accuracy under some time constraint. To achieve this goal, the confidence score threshold is adaptively adjusted based on the network condition, the confidence score and the selected frame resolution, before being used for determining if the video frame should be offloaded for further processing.
Extensive evaluation results show that the proposed solution can significantly outperform other approaches.

\section*{Acknowledgement}
This work was supported in part by the National Science Foundation under grant number 2125208.

\bibliographystyle{IEEEtran}
\bibliography{sample-base}

\begin{thebibliography}{10}
\providecommand{\url}[1]{#1}
\csname url@samestyle\endcsname
\providecommand{\newblock}{\relax}
\providecommand{\bibinfo}[2]{#2}
\providecommand{\BIBentrySTDinterwordspacing}{\spaceskip=0pt\relax}
\providecommand{\BIBentryALTinterwordstretchfactor}{4}
\providecommand{\BIBentryALTinterwordspacing}{\spaceskip=\fontdimen2\font plus
\BIBentryALTinterwordstretchfactor\fontdimen3\font minus
  \fontdimen4\font\relax}
\providecommand{\BIBforeignlanguage}[2]{{%
\expandafter\ifx\csname l@#1\endcsname\relax
\typeout{** WARNING: IEEEtran.bst: No hyphenation pattern has been}%
\typeout{** loaded for the language `#1'. Using the pattern for}%
\typeout{** the default language instead.}%
\else
\language=\csname l@#1\endcsname
\fi
#2}}
\providecommand{\BIBdecl}{\relax}
\BIBdecl

\bibitem{bbc-news}
``{Chinese police spot suspects with surveillance sunglasses},''
  \url{https://www.bbc.com/news/world-asia-china-42973456}.

\bibitem{alex-nips12}
A.~Krizhevsky, I.~Sutskever, and G.~E. Hinton, ``{Imagenet classification with
  deep convolutional neural networks},'' \emph{Advances in neural information
  processing systems (NIPS)}, 2012.

\bibitem{npu-gpu-compare}
``{HiSilicon Kirin 970 Performance Overview},''
  \url{https://www.anandtech.com/show/12195/hisilicon-kirin-970-power-performance-overview/6f}.

\bibitem{tan-infocom2020}
T.~Tan and G.~Cao, ``{FastVA: Deep Learning Video Analytics Through Edge
  Processing and NPU in Mobile},'' \emph{IEEE INFOCOM}, 2020.

\bibitem{teerapittayanon-icdcs17}
S.~Teerapittayanon, B.~McDanel, and H.~Kung, ``{Distributed deep neural
  networks over the cloud, the edge and end devices},'' \emph{IEEE ICDCS},
  2017.

\bibitem{wang-infocom2020}
S.~Wang, S.~Yang, and C.~Zhao, ``{SurveilEdge: Real-time Video Query based on
  Collaborative Cloud-Edge Deep Learning},'' \emph{IEEE INFOCOM}, 2020.

\bibitem{parkhi-bmvc15}
A.~Z. O.~Parkhi, A.~Vedaldi, ``{Deep Face Recognition},'' \emph{British Machine
  Vision Conference}, 2015.

\bibitem{lfw-dataset}
G.~B. Huang, M.~Ramesh, T.~Berg, and E.~Learned-Miller, ``{Labeled Faces in the
  Wild: A Database for Studying Face Recognition in Unconstrained
  Environments},'' University of Massachusetts, Amherst, Tech. Rep., 2007.

\bibitem{lapuschkin-cvpr16}
S.~Lapuschkin{, A. Binder, G. Montavon, K.-R. M{\"u}ller, and W. Samek},
  ``{Analyzing classifiers: Fisher vectors and deep neural networks},''
  \emph{IEEE Conference on Computer Vision and Pattern Recognition (CVPR)},
  2016.

\bibitem{everingham-ijcv15}
M.~Everingham{, S. Eslami, L. Van~Gool, C. Williams, J. Winn, and A.
  Zisserman}, ``{The Pascal Visual Object Classes Challenge: A
  Retrospective},'' \emph{Springer International Journal of Computer Vision
  (IJCV)}, 2015.

\bibitem{redmon-cvpr16}
J.~Redmon, S.~Divvala, R.~Girshick, and A.~Farhadi, ``{You only look once:
  Unified, real-time object detection},'' \emph{IEEE CVPR}, 2016.

\bibitem{lin-eccv14}
T.-Y. Lin, M.~Maire, S.~Belongie, J.~Hays, P.~Perona, D.~Ramanan,
  P.~Doll{\'a}r, and C.~L. Zitnick, ``{Microsoft coco: Common objects in
  context},'' \emph{European Conference on Computer Vision}, 2014.

\bibitem{guo-icml17}
C.~Guo{, G. Pleiss, Y. Sun, and K. Weinberger}, ``{On Calibration of Modern
  Neural Networks},'' \emph{ACM International Conference on Machine Learning
  (ICML)}, 2017.

\bibitem{niculescu-ICML05}
{A. Niculescu-Mizil and R. Caruana}, ``{Predicting Good Probabilities with
  Supervised Learning},'' \emph{ACM International Conference on Machine
  Learning (ICML)}, 2005.

\bibitem{zadrozny-sigkdd02}
B.~Zadrozny and C.~Elkan, ``{Transforming Classifier Scores into Accurate
  Multiclass Probability Estimates},'' \emph{ACM International Conference on
  Knowledge Discovery and Data Mining (SIGKDD)}, 2002.

\bibitem{hiai}
``{HiAI},'' https://developer.huawei.com/consumer/en/devservice/doc/\\2020315.

\bibitem{jia-mm14}
Y.~Jia, E.~Shelhamer, J.~Donahue, S.~Karayev, J.~Long, R.~Girshick,
  S.~Guadarrama, and T.~Darrell, ``{Caffe: Convolutional Architecture for Fast
  Feature Embedding},'' \emph{ACM International Conference on Multimedia},
  2014.

\bibitem{he-cvpr16}
K.~He, X.~Zhang, S.~Ren, and J.~Sun, ``{Deep residual learning for image
  recognition},'' \emph{IEEE CVPR}, 2016.

\bibitem{barbu-nips19}
A.~.Barbu{, D. Mayo, J. Alverio, W. Luo, C. Wang, D. Gutfreund, J. Tenenbaum,
  and B. Katz}, ``{Objectnet: A Large-Scale Bias-Controlled Dataset for Pushing
  the Limits of Object Recognition Models},'' \emph{Advances in neural
  information processing systems (NIPS)}, 2019.

\bibitem{luo-iccv17}
J.~Luo, J.~Wu, and W.~Lin, ``{Thinet: A Filter Level Pruning Method for Deep
  Neural Network Compression},'' \emph{IEEE ICCV}, 2017.

\bibitem{yu-cvpr18}
R.~Yu, A.~Li, C.~Chen, J.~Lai, V.~Morariu, X.~Han, M.~Gao, C.~Lin, and
  L.~Davis, ``{NISP: Pruning Networks Using Neuron Importance Score
  Propagation},'' \emph{IEEE CVPR}, 2018.

\bibitem{cuervo-mobisys10}
E.~Cuervo, A.~Balasubramanian, D.~Cho, A.~Wolman, S.~Saroiu, R.~Chandra, and
  P.~Bahl, ``{MAUI: making smartphones last longer with code offload},''
  \emph{ACM Mobisys}, 2010.

\bibitem{Geng-ICNP2015}
Y.~Geng, W.~Hu, Y.~Yang, W.~Gao, and G.~Cao, ``{Energy-efficient computation
  offloading in cellular networks},'' \emph{IEEE ICNP}, 2015.

\bibitem{Geng-infocom18}
Y.~Geng, Y.~Yang, and G.~Cao, ``{Energy-efficient computation offloading for
  multicore-based mobile devices},'' \emph{IEEE INFOCOM}, 2018.

\bibitem{han-mobisys16}
S.~Han, H.~Shen, M.~Philipose, S.~Agarwal, A.~Wolman, and A.~Krishnamurthy,
  ``{MCDNN: An Approximation-Based Execution Framework for Deep Stream
  Processing Under Resource Constraints},'' \emph{ACM Mobisys}, 2016.

\bibitem{tan-tmc21}
T.~Tan and G.~Cao, ``{Deep Learning Video Analytics Through Edge Computing and
  Neural Processing Units on Mobile Devices},'' \emph{IEEE Transactions on
  Mobile Computing}, To Appear.

\bibitem{chen-sensys18}
K.~Chen, T.~Li, H.-S. Kim, D.~E. Culler, and R.~H. Katz, ``{MARVEL: Enabling
  Mobile Augmented Reality with Low Energy and Low Latency},'' \emph{ACM
  Sensys}, 2018.

\bibitem{ran-infocom18}
X.~Ran, H.~Chen, X.~Zhu, Z.~Liu, and J.~Chen, ``{DeepDecision: A Mobile Deep
  Learning Framework for Edge Video Analytics},'' \emph{IEEE INFOCOM}, 2018.

\bibitem{tan-secon21}
T.~Tan and G.~Cao, ``{Deep Learning Video Analytics on Edge Computing
  Devices},'' \emph{IEEE SECON}, 2021.

\bibitem{liu-mobicom2019}
L.~Liu{, H. Li, and M. Gruteser}, ``{Edge Assisted Real-Time Object Detection
  for Mobile Augmented Reality},'' \emph{ACM International Conference on Mobile
  Computing and Networking}, 2019.

\bibitem{tan-ipsn21}
T.~Tan and G.~Cao, ``{Efficient Execution of Deep Neural Networks on Mobile
  Devices with NPU },'' \emph{ACM/IEEE International Conference on Information
  Processing in Sensor Networks (IPSN)}, 2021.

\bibitem{motamedi-esl19}
M.~Motamedi, D.~Fong, and S.~Ghiasi, ``{Cappuccino: efficient CNN inference
  software synthesis for mobile system-on-chips},'' \emph{IEEE Embedded Systems
  Letters}, 2019.

\bibitem{oskouei-mm16}
{Latifi Oskouei, Seyyed Salar and Golestani, Hossein and Hashemi, Matin and
  Ghiasi, Soheil}, ``Cnndroid: Gpu-accelerated execution of trained deep
  convolutional neural networks on android,'' \emph{ACM International
  Conference on Multimedia}, 2016.

\bibitem{huynh-mobisys17}
L.~N. Huynh, Y.~Lee, and R.~K. Balan, ``{Deepmon: Mobile gpu-based deep
  learning framework for continuous vision applications},'' \emph{ACM Mobisys},
  2017.

\end{thebibliography}

\end{document}